\documentclass[11pt, draftcls, onecolumn]{IEEEtran}
\IEEEoverridecommandlockouts %% Packages %%

\usepackage{tabularx}
\usepackage{booktabs}
\usepackage{graphicx}
\usepackage{subfigure}
\usepackage{color}
\usepackage{epsfig}
\usepackage{amssymb}
\usepackage{amsmath}
\usepackage{amsthm}
\usepackage{latexsym}
\usepackage{setspace}
\usepackage{bbm}
\usepackage{float}
\usepackage{dsfont}
\usepackage{caption}
\usepackage{enumerate}
\usepackage[ruled,vlined]{algorithm2e}

\newcommand{\dP}{\mathrm{P}}
\newcommand{\dQ}{\mathrm{Q}}
\newcommand{\bPP}[1]{{\dP_{#1}}}
\newcommand{\bQQ}[1]{{\dQ_{#1}}}
\newcommand{\bPr}[1]{{\mathbb{P}}\left(#1\right)}

\newcommand{\bP}[2]{\mathrm{P}_{#1}\left({#2}\right)}

\newcommand{\bEE}{{\mathbb{E}}}

\newcommand{\cA}{{\mathcal A}}
\newcommand{\cB}{{\mathcal B}}

\newcommand{\cE}{{\mathcal E}}

\newcommand{\cK}{{\mathcal K}}

\newcommand{\cM}{{\mathcal M}}

\newcommand{\mN}{{\mathbb N}}

\newcommand{\cO}{{\mathcal O}}
\newcommand{\cP}{{\mathcal P}}

\newcommand{\cR}{{\mathcal R}}

\newcommand{\cT}{{\mathcal T}}

\newcommand{\cV}{{\mathcal V}}

\newcommand{\bx}{\mathbf{x}}
\newcommand{\cX}{{\mathcal X}}
\newcommand{\by}{\mathbf{y}}

\newcommand{\ep}{\epsilon}

\newcommand{\ttlvrn}[2]{\left\| #1 - #2\right\|}

\newtheorem{theorem}{Theorem}

\newtheorem{corollary}[theorem]{Corollary}
\newtheorem*{corollary*}{Corollary}

\newtheorem{lemma}[theorem]{Lemma}
\newtheorem*{lemma*}{Lemmas}

\theoremstyle{remark}
\newtheorem*{remark*}{Remark}
\newtheorem*{remarks*}{Remarks}
\theoremstyle{definition}
\newtheorem{definition}{Definition}
\newtheorem{remark}{Remark}
\newtheorem{example}{Example}

\allowdisplaybreaks
 
\newenvironment{protocol}[1][htb]
  {%
   \begin{algorithm}[#1]%
  }{\end{algorithm}}

\newcommand{\bH}{{\mathbf H}}
\newcommand{\mH}{{\mathbb H}}
\newcommand{\oH}{\mathbb{H}}
\newcommand{\bR}{\mathbf{R}}

\newcommand{\rde}{\mathrm{RDE}}
\newcommand{\omn}{\mathrm{OMN}}

\newcommand{\Hsig}[1]{H_{\sigma_{#1}}}
\newcommand{\Rsig}[1]{R_{\sigma_{#1}}}

\newcommand{\Rco}[1]{\mathcal{R}_{\tt CO}\left(\displaystyle#1\right)} 
\newcommand{\Rcod}[1]{\mathcal{R}^{\Delta}_{\tt CO}\left(\displaystyle#1\right)} 
\newcommand{\dec}{\mathrm{DEC}}
\newcommand{\bRin}{\mathbf{R}^{\tt in}}
\newcommand{\Rin}[1]{{R}^{\tt in}_{#1}}
\newcommand{\bRout}{\mathbf{R}^{\tt out}}
\newcommand{\Rout}[1]{{R}^{\tt out}_{#1}}
\newcommand{\sigout}{\sigma^{\tt out}}

\newcommand{\Rstar}[1]{R^*_{#1}}

\newcommand{\err}{\mathtt{ERR}}
\newcommand{\nack}{\mathtt{NACK}}
\newcommand{\ack}{\mathtt{ACK}}

\newcommand{\hbX}{\widehat{\bf X}}

\begin{document}

\setlength\textfloatsep{0pt}

\title{Universal Multiparty Data Exchange and Secret Key Agreement}

\author{ \IEEEauthorblockN{Himanshu Tyagi$^\dag$} \and
  \IEEEauthorblockN{Shun Watanabe$^\ddag$} }

\maketitle

{\renewcommand{\thefootnote}{}\footnotetext{
%\hspace*{-.11in}\rule{24ex}{.05em}
\noindent$^\dag$Department of Electrical Communication Engineering,
Indian Institute of Science, Bangalore 560012, India.  Email:
htyagi@ece.iisc.ernet.in.

\noindent$^\ddag$Department of Computer and Information Sciences,
Tokyo University of Agriculture and Technology, Tokyo 184-8588, Japan.
Email: shunwata@cc.tuat.ac.jp.  }}

\maketitle

\renewcommand{\thefootnote}{\arabic{footnote}}
\setcounter{footnote}{0}

\begin{abstract}
Multiple parties observing correlated data seek to recover each
other's data and attain {omniscience}. To that end, they
communicate interactively over a noiseless broadcast channel -- each bit
transmitted over this channel is received by all the parties. 
We give
a universal interactive communication protocol, termed the {\em recursive data exchange protocol} (RDE), which
attains omniscience for any sequence of data
observed by the parties and provide an {individual sequence}
guarantee of performance. As a by-product, for observations of length
$n$, 
we show the universal rate optimality of RDE up to an
$\cO(n^{-1/2}\sqrt{\log n})$ term in a generative setting where the
data sequence is independent and identically distributed (in time). Furthermore, 
drawing on the duality between
omniscience and secret key agreement due to Csisz\'ar and Narayan, we obtain a universal protocol for generating a
multiparty secret key of rate at most $\cO(n^{-1/2}\sqrt{\log n })$
less than the maximum rate possible. A key feature of RDE
is its {\em recursive structure} whereby when a subset $A$ of parties
recover each-other's data, the rates appear as if the parties have been
executing the protocol in an alternative model where the parties in
$A$ are collocated.
\end{abstract} 

%%%%%%%%%%%%%%%%%%%%%%%%%%%%%%%%%%%%%%%%%%%%%%%
\section{Introduction}
An $m$ party omniscience protocol is an interactive communication
protocol that enables $m$ parties to recover each other's data. The
communication is error-free and is in a broadcast mode wherein the
transmission of each party is received by all the other parties. Such
protocols were first considered in \cite{OrlEl84b} in a two-party
setup, where bounds for the number of bits communicated on average and
in the worst-case were derived for the case when no error is
allowed. The $m$ party version, and the omniscience terminology, was
proposed in \cite{CsiNar04} where the collective observations of the
parties was assumed to be an independent and identically distributed
(IID) sequence generated from a known distribution\footnote{Throughout
  we shall restrict to finite random variables and use the phrase
  probability distribution interchangeably with probability mass
  function (pmf).} $\bPP{X_1\cdots X_m}$. It was shown in
\cite{CsiNar04} that a simultaneous communication protocol based on
sending random hash bits of appropriate rates attains the optimal
sum-rate $R\left(\bPP{X_1\cdots X_m}\right)$.  
A
common feature of these prior works is that the protocol relies on the
knowledge of the underlying distribution $\bPP{X_1\cdots X_m}$. Note that the protocol proposed in
\cite{CsiNar04} relies on the classic multiterminal source coding
scheme given in \cite{CsiKor80}. Thus, it inherits the
following universality feature from that scheme: If for $1\leq i \leq
m$ the $i$th party communicates rate $R_i$, the protocol attains
omniscience for any source distribution $\bPP{X_1\cdots X_m}$ for
which the rate vector $(R_1, \ldots, R_m)$ lies in the omniscience
rate-region corresponding to $\bPP{X_1\cdots X_m}$. Nevertheless, this
provides no guarantee of universal optimality for the sum-rate $(R_1+
\cdots +R_m)$ for an arbitrary source $\bPP{X_1\cdots X_m}$.

 A naive protocol entails using the first
$n^\prime$ samples to estimate the entropies 
involved and then applying the optimal
   protocol of \cite{CsiNar04} with rates satisying the entropy
   constraints. Specifically, by using the estimator for entropy
   proposed in \cite{WuY16}, we can estimate the entropy to within an
   approximation error of $\cO(1/\sqrt{n^\prime})$ using $n^\prime$ samples,
   where the constants implied by $\cO$ depend on the support size of
   the distribution. This results in an universally sum-rate optimality protocol, but for observations of length $n$, the overall excess rate of communication over
   the optimal rate is $\cO(n^\prime/n + 1/\sqrt{n^\prime})$, which is at best  $\cO(n^{-1/3})$. Furthermore, there is
   no guarantee of performance for this protocol for a fixed sequence
   $(\bx_1, ..., \bx_m)$ observed by the parties.

  In this paper, we
present a protocol for omniscience, termed the {\em recursive data exchange protocol} ($\rde$), that is {
  universal} and works for {individual sequences} of data in the spirit of \cite{ZivLem78}, namely
it attains omniscience with
probability close to $1$ for every specific data sequence.
 For a given sequence $(\bx_1, ..., \bx_m)$ of data consisting of $n$
 length observations, $\rde$ attains
   an excess communication rate of $\cO(n^{-1/2})$ over
   $R\left(\bPP{\bx_1\cdots \bx_m}\right)$ where $\bPP{\bx_1\cdots
     \bx_m}$ denotes the joint type of the
   observations. As a consequence, we show that for the generative model where the data of the parties is IID,
$\rde$ is universally sum-rate optimal
with an excess rate of $\cO(n^{-1/2}\sqrt{\log n})$.  Note that
even for the case when the underlying distribution is known, the
optimal rate can only be achieved asymptotically and an excess rate is
often needed. In particular, for\footnote{For $m >2$, a
  variant of $\rde$ is shown in \cite{TyaWat17} to attain the optimal second-order asymptotic term, which is $O(n^{-1/2})$, for worst-case rates when the distribution is known.} $m=2$, the precise leading asymptotic
term in excess worst-case rate was established in \cite{TVW15} and was shown to be $\cO(n^{-1/2})$.

An interesting application of $\rde$ appears in {\em
  secret key} (SK) agreement \cite{Mau93, AhlCsi93,
  CsiNar04}. Specifically, Csisz\'ar and Narayan showed in
\cite{CsiNar04} that an optimum-rate SK can be generated by first
attaining omniscience and then extracting secure bits from the
recovered data. We follow the same procedure here with $\rde$ in place of the omniscience protocol of \cite{CsiNar04} and
obtain a universal SK of rate at most $\cO(n^{-1/2}\sqrt{\log n})$
less than the optimal average and the worst-case rate. Note that for
the case $m=2$ with known distribution, the precise leading asymptotic
term in the gap to optimal worst-case rate was established in
\cite{HayTW14} and was shown to be $\cO(n^{-1/2})$.  Therefore, for
multiparty data exchange as well as SK agreement $\rde$ can
roughly attain the worst-case performance for the case of known
distributions, without requiring the knowledge of the
distribution. Also, for average rate, the universal
$O(n^{-1/2}\sqrt{\log n})$ gap to optimal rates attained by $\rde$ is to our knowledge the best-known. 

It was shown in \cite{YanH10} that interaction enables an $\ack-\nack$
based universal variable-length coding scheme for the Slepian-Wolf
problem, where only party $1$ needs to send its data to party $2$. Our
protocol, too, is interactive in a similar spirit, but it relies on
carefully increasing the rate of communication for each party.  Note
that while for $m=2$ a simple extension of the protocol in
\cite{YanH10} works for the data exchange problem as well, this is not
the case when $m>2$. For $m>2$, the order in which the parties
communicate must be carefully chosen. We give a very simple criterion
for choosing this order of communication and show that the resulting
protocol is universally rate-optimal. Specifically, the encoders in $\rde$ send random hash bits
  corresponding their inputs, while the decoders, which use a variant
  of minimum entropy decoding, try to decode the observations of
  any subset of communicating parties. A key feature of $\rde$
is its {\em recursive structure} whereby when a subset $A$ of parties
recover each-other's data, the rates appear as if the
parties have been executing the protocol in an alternative model where
the parties in 
$A$ are collocated from the start. To enable this, the parties communicate in the order of the entropies of
their empirical types, with the highest entropy party communicating
first, followed by the next highest entropy party, and so on. The
delay in communication between the parties is chosen to ensure that for every pair of communicating parties, the difference of their rates
of communication, at any instance, is equal to the difference of the entropies of their marginal types. We follow this policy and increase the rate in
steps until a subset of parties can attain {\em local omniscience}, $i.e.$, recover each other's data.

Our encoders are easy to implement, but the decoders
are theoretical constructs which use type classes to form a list of
guesses for the data of other parties. Furthermore, since we try to
decode the data of every possible subset of communicating parties, 
the complexity of our decoder is exponential in $m$. Nevertheless, 
we believe that $\rde$ is a stepping-stone towards a practical protocol for the multiparty data exchange problem.

There is a rich literature relating to the problems considered
here. Following the seminal work of Slepian and Wolf \cite{SleWol73},
which introduced fixed-length distributed source coding for two
parties, universal error-exponents for the multiparty extension of
this problem were considered in \cite{CsiKor81b, Csi82, OohHan94}. For
the case of two parties, universal {variable length} protocols with
optimal average rate were proposed in \cite{Draper04, YanH10}. In
particular, the protocol used in \cite{YanH10} has excess rate less
than\footnote{For $m=2$ even $\rde$ has excess rate
  less than $\cO(n^{-1/2})$. The extra $\cO(\sqrt{\log n})$ factor for
  a general $m$ appears since the optimal sum-rate may not be a
  concave function of $\bPP{X_1\cdots X_m}$ for $m>2$, and we take
  recourse to a Taylor approximation of the sum-rate function.}
$\cO(n^{-1/2})$, which is the best-known. A related protocol was used
in \cite{TVW15} in a single-shot setup which, when applied to IID
observations with a known distribution, was shown to be of optimal
worst-case length even up to the second-order asymptotic term. A
slight variant of the data exchange or omniscience problem, which assumes the data of
the parties to be elements of a finite field and requires exact
recovery using linear communication, has been considered in
\cite{RouSprSad10, SprSadBooRou10, CouXieWes10, MilPawRouGasRam,
  MilPawRouGasRam16}. While $\rde$ doesn't directly
relate to these works, we propose it as an alternative approach for
ensuring data exchange in these settings.

The remainder of this paper is organized as follows: The next section
contains the formal description of the omniscience problem. We first
describe an idealized version of $\rde$, $\rde_{\tt id}$, in
Section~\ref{s:protocol_description_ideal} where we assume that the
rates can be continuously increased and an ideal decoder is
available. We also illustrate the working of $\rde_{\tt id}$ with
examples. Ideal assumptions are removed in the subsequent section which
contains a complete description of $\rde$ and our
main results about its performance. The SK agreement problem and our universal SK agreement
protocol based on $\rde$ are described in Section~\ref{s:problem_description_SK}. All
the proofs are given in Section~\ref{s:technical}. Our proofs rely on
technical properties of the formula for minimum communication for
omniscience. Some of these properties are new and maybe of independent
interest.

{\em Notations.} We start by recalling the standard notations: We
consider discrete random variables $X$ taking values in a finite set
$\cX$ and with pmf $\bPP{X}$. {Denote the set $\{1, ...,
  m\}$ of all parties by $\cM$.} For random variables $(X_i : i \in
\cM)$ and $A\subseteq \cM$, $X_A$ denotes the collection $(X_i : i \in
A)$. Also, $X_A^n$ denotes the sequence of IID random variables
$\{X_{A,t}\}_{t=1}^n$, where $X_{A,t}= (X_{i,t} : i \in A)$.
Similarly, $\cX_A$ denotes the product set $\prod_{i \in A}\cX_i$ and
$\cX^n=\cX_1\times \cdots \times \cX_n$.  For given distributions
$\bPP{}$ and $\bQQ{}$, their variational distance is denoted by $\|
\bPP{} - \bQQ{}\| = \frac{1}{2}\sum_x | \bPP{}(x) - \bQQ{}(x)|$.
While our protocols are conceptually simple, the analysis is
notationally heavy and relies on some bespoke notations.  For easy
reference, we summarize all nonstandard notations used in this paper
in Table~\ref{t:notation}. We often need to think of a subset of
parties as a single party and use natural extensions of our notations
to indicate such cases. For instance, for a partition $\sigma$ of $A\subseteq \cM$ or of $\cM$,
the notation $\left(R^*_{\sigma_i}(A_\sigma) : 1\leq i\leq
|\sigma|\right)$ extends $R_i^*(A)$ given in Table~\ref{t:notation}
and denotes the solution $\left(R_1,\ldots, R_{|\sigma|}\right)$ for
equations
\[
\sum_{j \neq i}R_j = H\left(X_A| X_{\sigma_i}\right), \quad \forall\,
1\leq i\leq |\sigma|.
\]
{Note that we have abused the subscript notation, with
  different connotations in different contexts. For instance, we use the
 notation $A_\sigma$ for a
  partition $\sigma$ of $A$, which represents the set $A$ as a
  collection of elements $\sigma_i\in \sigma$. However, the specific
  connotation should be clear from the context.}

\begin{table}[h]
\centering
\begin{tabularx}{0.9\linewidth}{p{0.3\linewidth}Xp{0.6\linewidth}}
\toprule Notation& Description \\ \midrule $\Sigma(A)$& Set of all
nontrivial partitions of $A$ \\ $|\sigma|$& Number of parts in the
partition $\sigma$ \\ $\sigma_f(A)$& The finest partition $\{\{i\} :
i\in A\}$ of $A$ \\ $\sigma_B(A)$, $B\subsetneq A$& The partition
$\{\{A\setminus B\},\{i\} : i \in B\}$ of $A$ \\ $R_A$& Sum rate
$\sum_{i\in A}R_i$ \\ $\Rco{A}$& Set of all vectors $(R_i : i \in A)$
s.t. $R_B\geq H(X_B|X_{A\setminus B})$, $\forall\, B\subsetneq A$
\\ $\Rcod{A}$& Set of all vectors $(R_i : i \in A)$ s.t. $R_B\geq
H(X_B|X_{A\setminus B})+|B|\Delta$, $\forall\, B\subsetneq A$
\\ $R_{\tt CO}(A)$& Minimum of $R_A$ over all $\bR \in \Rco{A}$
\\ $\mH_{\sigma}(A)$, $\sigma\in \Sigma(A)$&
$\frac1{|\sigma|-1}\sum_{i=1}^{|\sigma|}H\left(X_{A}|
X_{\sigma_i}\right)$ \\ $R_i^*(A), i\in A$& Solution of $\sum_{j\neq
  i}R_j = H(X_A|X_i)$, $\forall\, i \in A$ \\ $A_\sigma$, $\sigma\in
\Sigma(A)$& $\{A_{\sigma_i} : 1\leq i \leq |\sigma|\}$, where
$A_{\sigma_i} = \sigma_i$ treated as a single party\\ $|\pi|$& Maximum
number of bits communicated in any execution of the protocol $\pi$
\\ $|\pi|_{\tt av}$& Expected value of the number of bits communicated
in an execution of the protocol $\pi$ \\ \bottomrule
\end{tabularx}
\caption{Summary of notations used in the paper.}
\label{t:notation}
\end{table}

%%%%%%%%%%%%%%%%%%%%%%%%%%%%%%%%%%%%%%%%%%%%%%%
\section{Omniscience}\label{s:problem_description_omniscience}
We begin with the description of the problem for IID
observations. Specifically, parties in a set $\cM = \{1, \ldots, m\}$
observe an IID sequence $X_\cM^n= \left(X_{\cM 1},\ldots, X_{\cM
  n}\right)$, with the $i$th party observing $\left\{X_{i
  t}\right\}_{t=1}^n$ and $X_{\cM t} = (X_{i t} : i \in \cM ) \sim
\bPP{X_\cM}$ denoting the collective data at the $t$th time
instance. The parties have access to shared public randomness (public
coins) $U$ such that $U$ is independent jointly of $X_\cM^n$.
Furthermore, the $i$th party, $i\in \cM$, has access to private
randomness (private coins) $U_i$ such that $U_\cM$, $U$, and $X_\cM^n$
are mutually independent.  Thus, the $i$th party observes $(X_i^n,
U_i, U)$.

For simplicity, we restrict our exposition to {\it tree-protocols}
($cf.$~\cite{KushilevitzNisan97}) described below.
%% However, our proposed protocol is rate-optimal even with a more
%% general class of protocols; see Appendix A. 
A tree-protocol $\pi$ for $\cM$ consists of a binary tree, termed the
{\it protocol-tree}, with the vertices labeled by the elements of
$\cM$. The protocol starts at the root and proceeds towards the
leaves. When the protocol is at vertex $v$ with label $i_v$, party
$i_v$ communicates a bit $b_v$ based on its local observations
$(X_{i_v}^n, U_{i_v}, U)$. The protocol proceeds to the left- or the
right-child of $v$, respectively, if $b_v$ is $0$ or $1$. The protocol
terminates when it reaches a leaf, at which point each party produces
an output based on its local observations and the bits communicated
during the protocol, namely the transcript $\Pi=\pi(X_\cM^n, U_\cM,
U)$.  Note that for tree-protocols the set of possible transcripts is
prefix-free. Also, note that the output is not included in
  the transcript of the protocol, but is computed locally at each
  party. The literature on distributed function computation often
  focuses on Boolean functions and includes the $1$-bit output as a
  part of the protocol transcript
  (cf.~\cite{KushilevitzNisan97}). This results in a negligible
  $1$-bit loss in communication. However, including the output in the
  transcript in our setup makes the data exchange problem trivial
  since the optimal protocol shall entail each party declaring its
  observation.

Figure~\ref{f:tree_protocols} shows an example of a protocol tree for
$m=3$. The label of each node represents the party which
  determines the communicated bit at that node; the final boxes
  represent the termination of the protocol, at which point an output
  is produced by each party.

\begin{figure}[h]
\centering \includegraphics[scale=0.3]{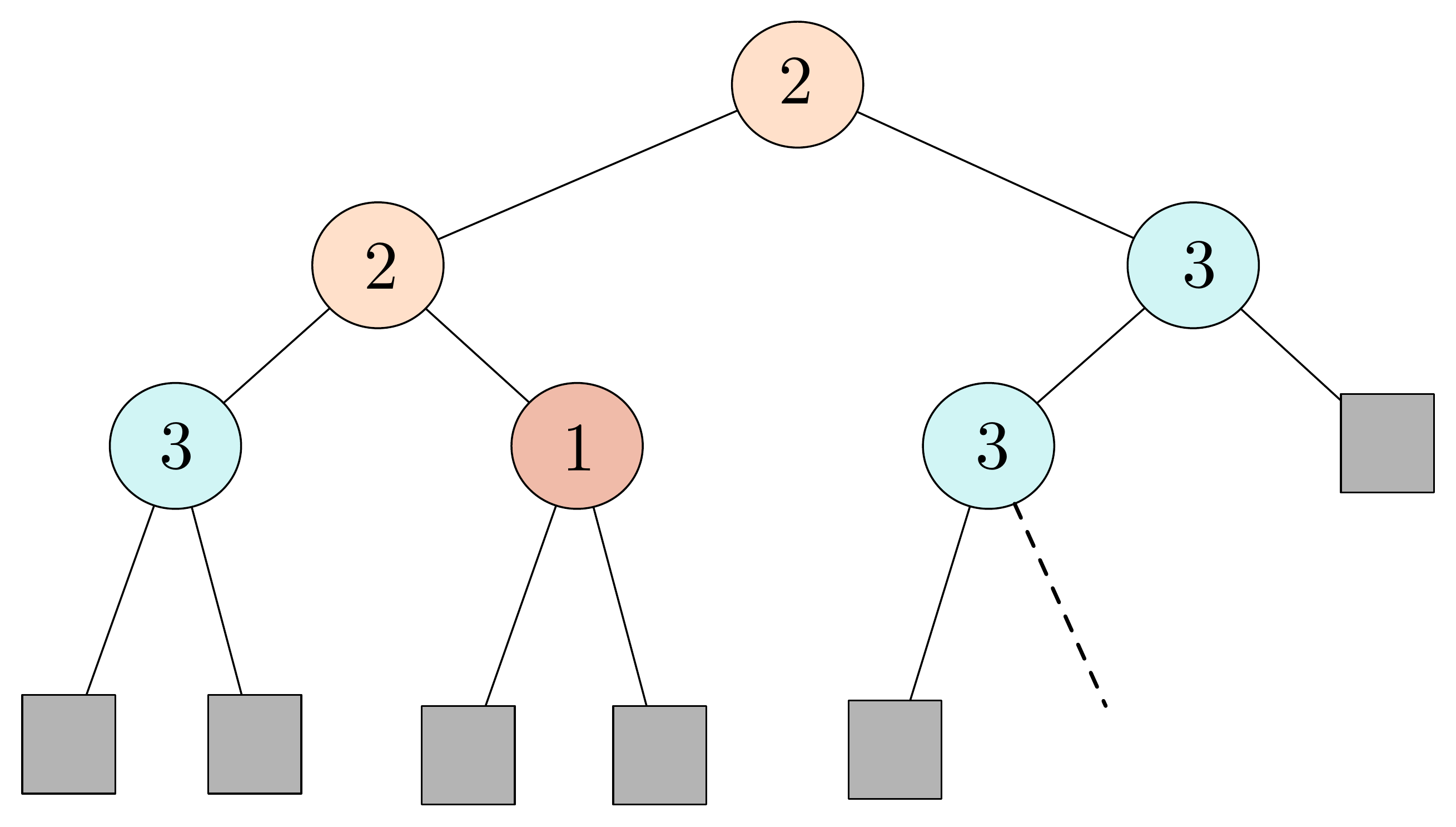}
\caption{A multiparty protocol tree.}
\label{f:tree_protocols}
\end{figure}

The (worst-case) length $|\pi|$ of a protocol $\pi$ is the maximum
number of bits that are transmitted in any execution of the protocol
and equals the depth of the protocol-tree.  Also, the average length
$|\pi|_{\tt av}$ is given by the expected value of the number of bits
transmitted in an execution of the protocol $\pi$.

In the omniscience problem, the parties engage in interactive
communication to recover each other's data. A protocol $\pi$
constitutes an $\ep$-omniscience protocol if, at the end of the
protocol, the $i$th party can output an estimate $\hbX_i=\hbX_i(X_i^n,
U_i, U, \Pi)\in \cX_\cM^n$ such that
\begin{align}
\bPr{\hbX_i= X_\cM^n : i \in \cM}\geq 1-\ep.  \nonumber
\end{align}
\begin{definition}[Communication for omniscience] Given  IID observations with a common distribution $\bPP{X_\cM}$ as above, for $0\leq  
  \ep<1$, a rate $R\geq 0$ is an $\ep$-achievable omniscience rate if
  there exists an $\ep$-omniscience protocol $\pi$ with length $|\pi|$
  less than $nR$, for all $n$ sufficiently large. The infimum over all
  $\ep$-achievable omniscience rates is denoted by
  $R_\ep(\bPP{X_\cM})$. The {\it minimum rate of communication for
    omniscience} $R(\bPP{X_\cM})$ is given by
\begin{align}
R(\bPP{X_\cM}) = \lim_{\ep \to 0} R_\ep(\bPP{X_\cM}).  \nonumber
\end{align}

The {\it minimum average rate of communication for omniscience}
$R^{\tt av}(\bPP{X_\cM})$ is defined similarly by replacing length
$|\pi|$ with average length $|\pi|_{\tt av}$.
\end{definition}
{The fundamental quantity $R(\bPP{X_\cM})$ was characterized in \cite{CsiNar04} as 
\begin{align}
R(\bPP{X_\cM}) 
&= \min\left\{ \sum_{i=1}^m R_i : \sum_{i\in B}R_i \geq H(X_B|X_{B^c}), \quad \forall\, B \subsetneq\cM \right\}.
\label{e:rco_formula-LP}
\end{align}
%\begin{remark}[Rate-region of communication for omniscience]\label{r:rate_region} 
Following \cite{CsiNar04}, the collection of all rate vectors $\bR =
(R_1,\ldots,R_m)$ satisfying the constraints in
\eqref{e:rco_formula-LP}, termed the CO region, will be denoted by
$\Rco{\cM|\bPP{X_\cM}}$, and the minimum sum-rate by $R_{\tt
  CO}\left(\cM|\bPP{X_\cM}\right)$. When the distribution
$\bPP{X_\cM}$ is clear from the context, we shall omit it from the
notation and simply use $\Rco{\cM}$ and $R_{\tt CO}\left(\cM\right)$.}
%\end{remark}

{While the result in \cite{CsiNar04} was shown to hold only for
$R(\bPP{X_\cM})$, the same characterization holds for $R^{\tt
  av}(\bPP{X_\cM})$ as well. Indeed, note that the set of distinct transcripts of a tree protocol $\pi$ is prefix-free. Therefore, the lengths of these transcripts satisfy Kraft's inequality, and so, $H(\Pi)
\leq |\pi|_{\tt av}$.  By proceeding exactly as in \cite{CsiNar04}, we can see that 
$R^{\tt av}(\bPP{X_\cM}) \geq R_{\tt
  CO}(\cM|\bPP{X_\cM})$. On the other hand, clearly $R^{\tt
  av}(\bPP{X_\cM})\leq R(\bPP{X_\cM}) = R_{\tt CO}(\cM
|\bPP{X_\cM})$, whereby for every distribution $\bPP{X_\cM}$, we have
\[
R^{\tt av}(\bPP{X_\cM}) = R_{\tt CO}(\cM
|\bPP{X_\cM}).
\]
An alternative expression
for $R_{\tt CO}(\bPP{X_\cM})$ was obtained in \cite{CsiNar04} by looking at its dual
form. In fact, by leveraging on the complementary slackness property,
\cite{Cha08, ChaZhe10} showed that the optimization in   
the dual form can be restricted to the partitions of $\cM$ and showed   
that\footnote{An alternative proof of \eqref{e:rco_formula} was provided in 
\cite{ChanBEKL15} by using techniques from submodular
optimization.}
\begin{align}
R_{\tt CO}(\cM|\bPP{X_\cM})&= \max_{\sigma\in \Sigma(\cM)}
\mH_{\sigma}(\cM|\bPP{X_\cM}),
\label{e:rco_formula}
\end{align}
where $\Sigma(\cM)$ denotes the set of partitions of $\cM$, and, for
each $\sigma\in \Sigma(\cM)$,
\begin{align}
\mH_{\sigma}(\cM|\bPP{X_\cM}) =
\frac1{|\sigma|-1}\sum_{i=1}^{|\sigma|}H\left(X_{\cM}|
X_{\sigma_i}\right).
\label{e:H_sigma_definition}
\end{align}
Note that the fact that $R_{\tt CO}(\cM|\bPP{X_\cM})$ is lower bounded by the right-side of \eqref{e:rco_formula} was shown earlier in \cite{CsiNar04}.  
$\rde$ directly achieves the right-side of \eqref{e:rco_formula}, thereby providing an
alternative, ``operational" proof for the tightness of this lower
bound for $R_{\tt CO}(\cM|\bPP{X_\cM})$ from \cite{CsiNar04}.}

While there can be several maximizers of $\mH_\sigma$, there exists a
maximizing partition which is a further partition of any other
maximizing partition \cite[Theorem 5.2]{ChanBEKL15}, the finest
maximizing partition; we shall call this finest maximizer of
$\mH_\sigma$ in \eqref{e:rco_formula} the {\it finest dominant
  partition} (FDP), which was called {\em fundamental partiion} in \cite{ChanBEKL15}. The {\it finest partition} $\sigma_f(\cM) :=
\left\{\{i\}, i\in\cM\right\}$ plays a particularly important role in
$\rde$. Note that when the finest partition is FDP, the optimal
rate assignment is uniquely given by the solution $\bR^*=
(R_1^*,\ldots,R_m^*)$ of
\begin{align} 
\sum_{i \in \cM\backslash \{j \}} R_i = H(X_\cM |
X_j),~~~j=1,\ldots,m.
\label{e:R_star_definition}
\end{align}
%%%%%%%%%%%%%%%%%%%%%%%%%%%%%%%%%%%%%%%%%%%%%%%
\section{Universal protocol for omniscience under ideal assumptions}\label{s:protocol_description_ideal}
We give a universal protocol for omniscience, which, when a sequence
$\bx_\cM$ is observed, will transmit communication of rate no more
than $R_{\tt CO}\left(\cM| \bPP{\bx_\cM}\right)$.  To present the main
idea behind $\rde$, we first describe it assuming the following ideal assumptions.

Specifically, we make two assumptions:
\begin{enumerate}
\item[(a)] {\it Continuous rate assumption:} Communication-rate, defined as 
the total number of bits of communication up to a certain
  time divided by $n$, can be increased continuously in time\footnote{Clearly, this does not hold in
    practice since the number of bits of communication can be increased only in steps of discrete sizes. The continuous rate assumption allows us to examine,
    loosely speaking, the ``fluid limit'' behavior of $\rde$.}; and
\item[(b)] {\it Ideal decoder assumption:} We assume the availability
  of an error-free, ideal decoder $\dec_{\tt id}$ which correctly
  decodes a sequence once sufficient communication has been sent and
  declares a $\nack$ otherwise.\footnote{In analysis of the ideal
    protocol, we do not account for the rate needed to send
    $\nack$s. In practice, each $\nack$ symbol counts for a bit of
    communication and the size $\Delta$ of discrete increments must be
    chosen carefully to render the rate needed to send $\nack$s
    negligible.}
\end{enumerate} 
A standard universal decoder used in source coding is the minimum
entropy decoder which, given side-information $\by$ and an $nR$-bit\footnote{$nR$ is required
  to be an integer. When this is not the case, we simply use $\lceil
  nR \rceil$ bits in place of $nR$. This convention will be used
  throughout this paper and will be accounted for in our analysis.}
random hash\footnote{A ``random hash" of $X^n$ is a bit
    sequence produced by a function $f:{\cal X}^n \to \{0,1\}^{nR}$
    which is chosen randomly (using public randomness) from a class of
    functions satisfying the $2$-universal property \cite{CarWeg79}.
    For instance, the class of all functions satisfies the
    $2$-universal property and, therefore, standard ``random binning"
    (cf. \cite{CovTho06}) produces a random hash.}  of
$X^n$, searches for the unique sequence $\bx$ such that the joint type
$\bPP{\overline{X}\,\overline{Y}}=\bPP{\bx\by}$ satisfies
$H\left(\overline X \middle| \overline Y\right)\leq R$ and the hash of
$\bx$ matches the received hash bits. The decoder that we prescribe in
the next section works on a similar principle except that it searches
for any possible subset of sequences it can decode with the current
rate. To avoid the additional complications due to decoding error, we
first assume the availability of an ideal decoder $\dec_{\tt id}$
which enables omniscience for all parties $j\in A$ as soon as the rate
received from the parties in $A$ is sufficient. That is, the ideal
decoder guarantees that each party $i \in A$ can recover the correct
sequence $\bx_A$ if the rates of communication $\bR = (R_i : i \in
\cA)$ satisfy $\bR \in \Rco{A|\bPP{\bx_A}}$. Furthermore, if
$\bR\notin\Rco{A|\bPP{\bx_A}}$, the ideal decoder does not mistakenly
output a wrong sequence $\bx^\prime_A$, but declares a $\nack$
instead. Protocol~\ref{p:dec_id} summarizes our assumed ideal decoder $\dec_{\tt id}$.
%%%%%%%%%%%%%%%%%%%%%%%%%%%%%%%%%%%%%%%%%%
%%%%%%%%%%%%%%%%%
\begin{protocol}[h]
\caption{Ideal decoder $\dec_{\tt id}(j, \sigma, \bR)$} \KwIn{An index
  $1\leq j\leq m$, a partition $\sigma\in \Sigma(\cM)$, a rate vector
  $\bR=(R_1, \ldots, R_m)$.}  \KwOut{An $\ack$ message $(\ack, A)$ or
  a $\nack$ message}
\begin{enumerate}
\item For $\sigma_i$ such that $j\in\sigma_i$, search for the maximal
  set $A\subseteq \cM$ such that $\sigma_i\subsetneq A$ and $(R_l :
  l\in A) \in \Rco{A\mid \bPP{\bx_A}}$, and reveal $\bx_A$ to party
  $j$.
\item \uIf{If such an $A$ was found in Step 1}{return $(\ack,A)$.}
  \Else{return $\nack$.}
\end{enumerate}
\label{p:dec_id}
\end{protocol}
%%%%%%%%%%%%%%%%%

With this ideal decoder at our disposal, under the continuous rates
assumption, finding a universal protocol is tantamount to finding a
policy for increasing the rates $(R_1, \ldots, R_m)$ such that when
the rate vector enters $\Rco{\cM| \bPP{\bx_\cM}}$ for the first time,
the sum-rate is $R_{\tt CO}\left(\cM| \bPP{\bx_\cM}\right)$.  Note
that initially the marginal types $\bPP{\bx_i}$ are available to each
party and can be transmitted using $\cO(\log n)$ bits, since there are
only polynomially many types. Also, if a subset $A$ attains local
omniscience in the middle of the protocol, any $j\in A$ upon
recovering $\bx_A$ can transmit $\bPP{\bx_A}$ in $\cO(\log n)$ bits to
all the parties, who in turn can use it to compute $H(\bPP{\bx_A})$.

As an illustration, consider the simple case when $m=2$. Parties first
share $\bPP{\bx_1}$ and $\bPP{\bx_2}$; suppose $H(\bPP{\bx_1}) \geq
H(\bPP{\bx_2})$. Then, party $1$ starts communicating and increases
its rate $R_1$ at slope\footnote{The slope is defined as the
  derivative of rate with respect to the time under the continuous
  rate assumption.} $1$. When the rate $R_1$ reaches $H(\bPP{\bx_1}) -
H(\bPP{\bx_2})$, party $2$ starts communicating at slope $1$ as
well. Throughout the protocol, each party is trying to decode the
other using the ideal decoder $\dec_{\tt id}$ and they keep on
communicating as long as the ideal decoders output $\nack$s. The
parties will decode each other as soon as $(R_1, R_2)$ enters
$\Rco{\{1,2\}|\bPP{\bx_1,\bx_2}}$, $i.e.$, when
\begin{align}
R_1\geq H(\overline{X}_1|\overline{X}_2)\text{ and } R_2\geq
H(\overline{X}_2|\overline{X}_1), \nonumber
\end{align}
where $(\overline{X}_1, \overline{X}_2)\sim \bPP{\bx_1, \bx_2}$. Note
that once both parties start communicating, the difference $R_1-R_2$
is maintained as $H(\overline{X}_1)-H(\overline{X}_2)$. Thus, when
$(R_1, R_2)$ enters $\Rco{\{1,2\}}$, it holds that
\begin{align*}
R_1= H(\overline{X}_1| \overline{X}_2) \text{ and } R_2=
H(\overline{X}_2| \overline{X}_1);
\end{align*}
the red line in Figure~\ref{f:ex_2party} illustrates\footnote{It is
  also possible to proceed along the blue line for the $m=2$
  case. However, its extension to a general $m$ is not clear.} this
evolution of rates.

\begin{figure}[t]
\centering \includegraphics[scale=0.3]{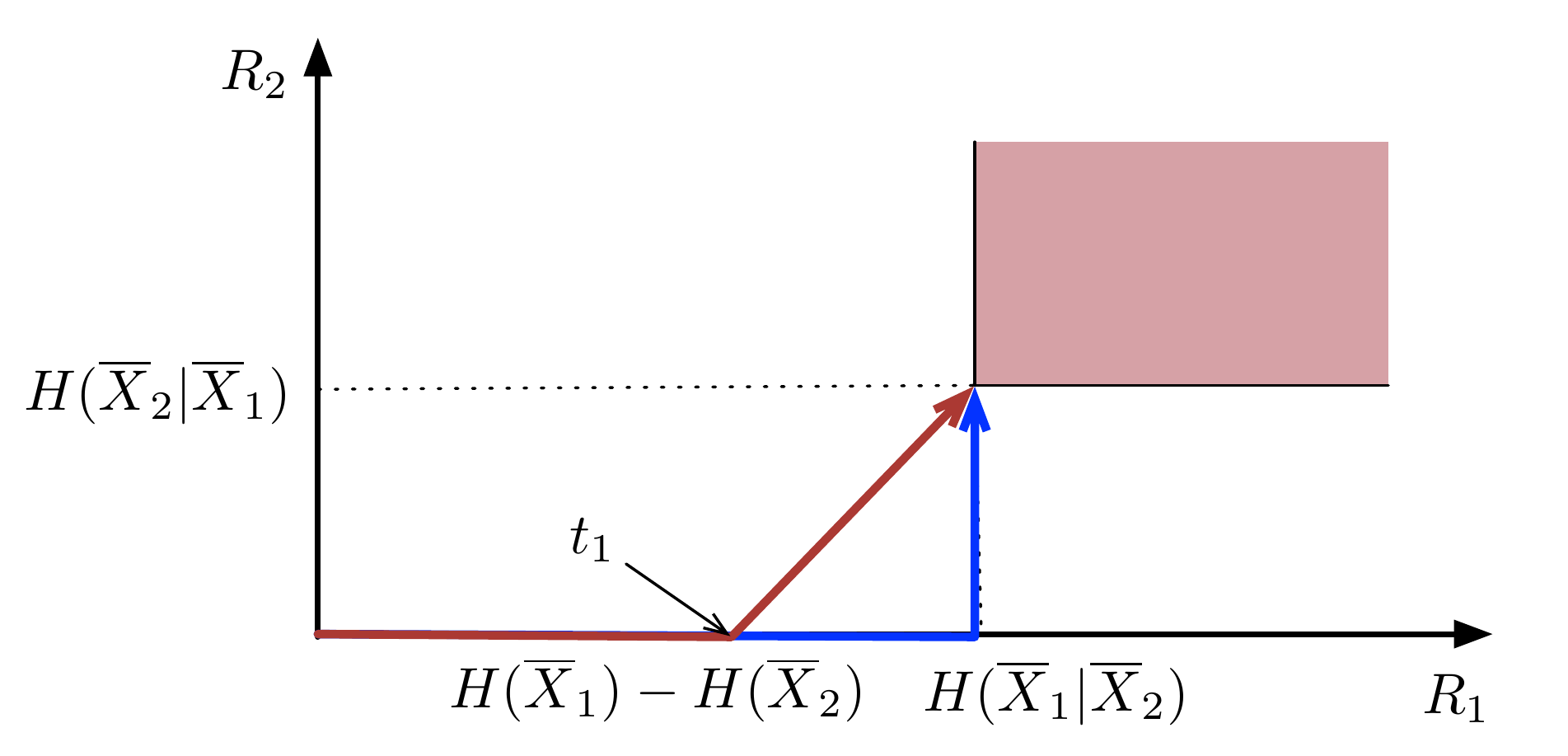}
\caption{Illustration of protocol for $m=2$. The transition point
  $t_1$ depends only on the marginal types $\bPP{\bx_1}$ and
  $\bPP{\bx_2}$. }
%and can be shared using
%  $\cO(\log n)$ bits.
\label{f:ex_2party}
\end{figure}

$\rde$ extends the idea above to a general $m$. We design $\rde$ so that the first subset $A$ which attains local omniscience
does so by using communication only from the parties in $A$ and of sum
rate
\begin{align}
R_A = \mH_{\sigma_f(A)}(A|\bPP{\bx_A}) = \sum_{i\in A}R_i^*(A);
\label{e:star_sum}
\end{align}
see \eqref{e:rstar_sum_formula2} in Lemma~\ref{l:rstar_char} given in
Section~\ref{s:technical} below for the second equality. To that end,
we note (see Lemma~\ref{l:rstar_char} for a proof) that for every $A$
\begin{align}
R_i^*(A) - R_j^*(A)=H(\overline{X}_i)-H(\overline{X}_j).
\label{e:R_star_diff}
\end{align}
A key point here is that for $\bPP{\bx_\cM}$ this difference can be
computed using only the marginal types $\bPP{\bx_i}$ and
$\bPP{\bx_j}$. $\rde$ ensures that for every pair $(i,j)$ of
communicating parties, the rate of communication
\[
R_i - R_i^*(A) = R_j - R_j^*(A),
\]
which by \eqref{e:R_star_diff} in turn can be ensured if the {\it constant difference}
property, namely
\begin{align}
R_i - R_j = H(\overline{X}_i) - H(\overline{X}_j),
\label{e:constant_diff}
\end{align}
is maintained throughout the protocol for every pair of communicating
parties.  Thus, all communicating parties $i$ reach the rate
$R_i^*(A)$ at the same time. Specifically, we first arrange parties in
decreasing order of the entropy of the empirical distribution of their
local observations, which are shared in $\cO(\log n)$-bits.  Assuming
$H(\bPP{\bx_1})\geq H(\bPP{\bx_2})\geq \cdots \geq H(\bPP{\bx_m})$,
party $1$ starts communicating, and the $i$th party starts
communicating when $R_1\geq H(\bPP{\bx_1}) - H(\bPP{\bx_i})$. This
ensures the constant difference property \eqref{e:constant_diff} for
every pair $(i,j)$ of communicating parties.  {For
  notational convenience, we assign $-1$ to $R_i$ when the $i$th party has not started communicating; the rate vector
  $(0,-1,-1,\ldots,-1)$ indicates that party 1 starts communicating
  and every one else remains quiet.}  When a subset $A$ attains local
omniscience, we decrease the rate-slope for each party $i\in A$ to
$1/|A|$, thereby ensuring that collectively parties in $A$ increase
the rate of communication $R_A$ at slope $1$. Note that since parties
in $A$ have recovered $\bx_A$, any one party $i \in A$ can compute the
type $\bPP{\bx_A}$ and transmit it using $\cO(\log n)$ bits. Our main
observation is that at this point the rates appear as if the parties
in $A$ were collocated to begin with and have been executing the
protocol as a single party. In particular, $R_A - R_j =
H(\overline{X}_A) - H(\overline{X}_j)$ for any communicating party $j$ outside $A$.  The second crucial observation is that for the first
subset $A$ which attains local omniscience, $(R_i^*(A) : i \in A)\in
\Rco{A}$. Since by \eqref{e:star_sum} $\sum_{i\in A}R_i^*(A)$ is a lower bound for
$\Rco{A}$, the parties in $A$ cannot attain local omniscience before they
communicate at sum-rate $\sum_{i\in A}R_i^*(A)$. Further, $\rde$ 
ensures that all parties in $A$ reach the rate $R_i^*(A)$ at the same
time. Thus, the parties in $A$ must have communicated at
sum-rate
\begin{align}
R_A = \sum_{i\in A}R_i^*(A) = \mH_{\sigma_f(A)}(A|\bPP{\bx_A})
\label{e:R_star_sum}
\end{align}
when they attain local omniscience. As the protocol proceeds,
subsets of parties keep attaining local omniscience and start behaving
as a single party. Proceeding recursively, it follows that when all
parties attain omniscience, the rate of communication must equal
$\mH_\sigma(\cM|\bPP{\bx_\cM})$ for some $\sigma \in \Sigma(\cM)$,
which in view of \eqref{e:rco_formula} is no more than $R_{\tt
  CO}(\cM| \bPP{\bx_\cM})$ and must be optimal in the limit as
$n\rightarrow \infty$.

To help the reader build heuristics for the complete protocol and its
analysis, we provide a sketch of the analysis for the ideal
situation and consider the ideal version $\rde_{\tt id}$. The formal proofs for the ideal case closely follow those 
for the results for the actual protocol and have been omitted. As
mentioned, $\rde_{\tt id}$ proceeds recursively by increasing the
rates with fixed slopes until a subset attains omniscience, at which
point the slopes are changed so that the parties in an omniscience
attaining subset behave as if they are collocated. We describe the
one-step omniscience protocol $\omn_{\tt id}$ in
Protocol~\ref{p:omn_ideal}. The protocol takes as input a partition
$\sigma$ such that parties in any one part are behaving as collocated
parties, a vector $\bH = \left(H_{\sigma_i}, 1\leq i\leq
|\sigma|\right)$ consisting of estimates of entropy for marginal
distribution of parties in any part of $\sigma$, and a rate vector
$\bR=(R_1, \ldots, R_m)$ of rates of communication sent by all the
parties up to this point.
 
%%%% 
\begin{protocol}[t]
\caption{$\omn_{\tt id}(\sigma, \bH,\bR)$} \KwIn{A partition
  $\sigma\in \Sigma(\cM)$ with $|\sigma|=k$, an entropy estimate
  vector $\bH = \left(H_{\sigma_i} : 1\leq i\leq k\right)$, a rate
  vector $\bR=(R_1, \ldots, R_m)$; we assume that $\bH$ is sorted,
  $i.e.$, $H_{\sigma_1}\geq H_{\sigma_2}\geq \cdots \geq
  H_{\sigma_k}$.}  \KwOut{A rate vector $\bRout$, a family of subsets
  $\cO$ that have attained omniscience.}

\begin{enumerate}
\item Initialize {$s:= \max\{i: R_{\sigma_i} \geq 0\}$}.
\item All parties $j$ such that $j\in \sigma_i$ for some
  {$1 \leq i\leq s$} increase their rates $R_j$ at slope
  $1/|\sigma_i|$.
\item \uIf{There exists $i>s$ such that $R_{\sigma_1}\geq
  H_{\sigma_1}- H_{\sigma_i}$}{ set $R_j = 0$ for all $j\in \sigma_i$,
  and set {$s= \max\{i: R_{\sigma_i} \geq 0\}$}.  }
\item For all $j$ such that $j\in \sigma_i$ for some $1\leq i \leq s$,
  execute $\dec_{\tt id}(j, \sigma, \bR)$, which outputs $\nack$ or
  $(\ack, A_j)$.
\item \uIf{All parties send a $\nack$} { return to Step 2.  }
  \Else{Identify the omniscience family
\[
\cO = \{B\subset\cM: \text{ all } j\in B \text{ returned } (\ack,
B)\}.
\] 
Set $\bRout = \bR$ and return $(\bR, \cO)$.  }
\end{enumerate}
\label{p:omn_ideal}
\end{protocol}
%%%%%

Note that a ``valid'' rate vector should reflect that parties in any
one part have communicated enough to attain local omniscience. Also,
since we shall recursively call $\omn_{\tt id}$, the only rate vectors
$\omn_{\tt id}$ encounters are those which can arise by increasing the
rates in the manner of $\rde$. We call the collection of rate
vectors satisfying the two conditions above $(\sigma,
\bH)$-valid. Formally,
%%%%
\begin{definition}\label{d:valid_id}
For $\sigma\in \Sigma(\cM)$ with $|\sigma|=k$ and $\bH = (\Hsig 1,
\ldots, \Hsig k)$ with $\Hsig 1\geq \Hsig 2\geq \cdots \geq \Hsig k$,
a rate vector $(R_1, \ldots, R_m)$ is $(\sigma, \bH)$-valid if
\[
(R_j, j\in\sigma_i) \in \Rco{\sigma_i},\quad \forall\, i \text{
  s.t. }|\sigma_i|\geq 2,
\]
and $(R_{\sigma_i}, 1\leq i\leq k)$ can be obtained by starting with
$(0, -1, -1, \ldots, -1)$ and incrementing the rates as in Protocol
\ref{p:omn_ideal} when the parties in each part $\sigma_i$ are
collocated, $i.e.$, each part $\sigma_i$ starts increasing its rate at
slope $1$ once $R_{\sigma_1} \geq \Hsig 1 - \Hsig i$.
\end{definition}
{As mentioned earlier, instead of initializing all rates with $0$ in $\rde$, and in the definition of a valid rate vector, we
  distinguish between rate $0$ and rate $-1$ for a technical reason. A
  rate of $-1$ indicates that the party is not participating in the
  protocol yet and will not even attempt to decode. In contrast, a $0$
  rate indicates that the party has not yet communicated any bits, but
  has started decoding and will increment its communication rate in
  each step from here on.}

%%%%
The result below shows a recursive property of $\omn_{\tt id}$ that
renders $\rde$ universally rate-optimal. Specifically, it
shows that if $\bR$ is $(\sigma,\bH)$-valid then, when $\omn_{\tt
  id}(\sigma, \bH,\bR)$ terminates, the output rate vector is
$(\sigma^{\tt out}, \bH^{\tt out})$-valid where $\sigma^{\tt out}$
{is a sub-partition of $\sigma$ which is obtained by
  combining the parts that have achieved local omniscience; $\bH^{\tt
    out}$ is the corresponding estimate for entropies of the marginals
  of parts of $\sigma^{\tt out}$.} Furthermore, for every set $A$ that
attains local omniscience, the sum-rate $R_A$ at the end of $\omn_{\tt
  id}$ is exactly $\mH_{\sigma_f(A_\sigma)}(A_\sigma)$.\footnote{When
  $A = \cup_{l=1}^c \sigma_{i_l}$, by our convention
  $\mH_{\sigma_f(A_\sigma)}(A_\sigma) =
  \mH_{\{\sigma_{i_1}|\cdots|\sigma_{i_c}\}}\left(A|\bPP{\bx_A}\right)$.}
%%%
\begin{theorem}\label{t:recursion_id}
For $\sigma\in \Sigma(\cM)$ with $|\sigma|=k$ and $\bH = (\Hsig 1,
\ldots, \Hsig k)$ with $\Hsig 1 \geq \Hsig 2 \geq \cdots \Hsig k$, let
$\bRin = (\Rin 1, \ldots,\Rin m)$ be $(\sigma,\bH)$-valid. Then, if
$\omn_{\tt id}(\sigma,\bH,\bRin)$ is executed, the final rates
$\bRout$ and the omniscience family $\cO$ satisfy the following:
%\begin{enumerate}
%\item 

1) Every $A\in \cO$ consists of parts of $\sigma$, $i.e.$,
\[
A = \bigcup_{l=1}^c\sigma_{i_l}
\]
for some $\{i_1,\ldots,i_c\} \subseteq \{1,\ldots,|\sigma|\}$, and the
sum-rate $\Rout A$ satisfies
\begin{align}
\Rout A =
\mH_{\{\sigma_{i_1}|\cdots|\sigma_{i_c}\}}\left(A|\bPP{\bx_A}\right).
\nonumber
\end{align}

%\item 
2) Let $\sigout\in\Sigma(\cM)$ be the partition obtained by combining
the parts in $\sigma$ that belong to the same $A$ in $\cO$.  Let
$H_{\sigout_i}$ denote the entropy of the type of
$\bx_{\sigout_i}$. Then, with $\bH^{\tt out}= \left(H_{\sigout_i},
1\leq i \leq |\sigout|\right)$, $\bRout$ is $(\sigout, \bH^{\tt
  out})$-valid.
%\end{enumerate}
\end{theorem}
In fact, Theorem~\ref{t:recursion_id} is a special case of
Theorem~\ref{t:recursion}, and the proof of the former follows from
that of the latter given below. {However, we provide a
  brief sketch of the proof of Theorem~\ref{t:recursion_id} here to
  highlight the key ideas and, also, to clarify the technical proof of
  Theorem~\ref{t:recursion}.}

{\noindent {\it Proof sketch.} For simplicity, assume that
  $\sigma$ consists of singletons, $i.e.$, $\sigma =
  \sigma_f(\cM)$. The main component of our proof is the following
  claim:} 
{\begin{center} {\em {\bf Claim:} The parties in
      a subset $A$ attain local omniscience exactly when each $R_i$,
      $i\in A$, reaches $\Rstar {i}(A)$.}
\end{center}
As mentioned before, all communicating parties $i\in A$ reach
$\Rstar{i}(A)$ simultaneously, and the parties in $A$ cannot attain
local omniscience before this happens.  The proof of the claim follows
from Lemma~\ref{lemma:induction} given in Section~\ref{s:technical},
since no subset of $A$ has attained local omniscience before $A$. }

{The theorem follows. Indeed, the first assertion holds by
  \eqref{e:R_star_sum}.  For the second assertion, we need to show
  that for two subsets $A$ and $B$ in $\cO$, $\Rout A - \Rout B =
  H(X_A) - H(X_B)$.  The complete proof considers various cases
  depending on if $A$ (or $B$) contains a party $i$ with nonnegative
  $\Rin i$. We illustrate the proof for a case when there exist $i\in
  A$ and $j\in B$ with $\Rin i, \Rin j\geq 0$. Since $\bRin$ is valid
  for $\sigma = \sigma_f(\cM)$ and the communicating parties maintain
  the difference of their rates, it follows from the claim above that
\begin{align*}
\Rout A - \Rout B &= \Rout {A\setminus \{i\}} - \Rout {B\setminus
  \{j\}} + \Rout i - \Rout j \\ &= \Rout {A\setminus \{i\}} - \Rout
      {B\setminus \{j\}} + \Rin i - \Rin j \\ &= \Rout {A\setminus
        \{i\}} - \Rout {B\setminus \{j\}} + H(X_i) - H(X_j) \\ &=
      \sum_{l\in A\setminus \{i\}} \Rstar{l}(A) - \sum_{k\in
        B\setminus \{j\}} \Rstar{k}(B) + H(X_i) - H(X_j) \\ &=H(X_A|
      X_i) - H(X_B| X_j) + H(X_i) - H(X_j) \\ &= H(X_A) - H(X_B).
\end{align*}
Other cases can be handled similarly. Therefore, $\bRout$ is valid for
$\sigout$.  \qed}

Thus, if we proceed by recursively calling $\omn_{\tt id}$, each time
with $(\sigma^{\tt out}, \bH^{\tt out}, \bR^{\tt out})$ obtained from
the previous call, we shall ultimately attain omniscience using the
sum rate $\mH_\sigma(\cM)$ for some partition $\sigma$. Since
$\mH_\sigma(\cM)$ is a lower bound for $\Rco{\cM}$ by \eqref{e:rco_formula}, this rate must be
optimal. We summarize the overall ideal
protocol in Protocol~\ref{p:main_id}.
%%%%%%
\begin{protocol}[h]
\caption{$\rde_{\tt id}$: The recursive data exchange protocol under ideal conditions}
\begin{enumerate}
\item Initialize $\sigma= \sigma_f(\cM)$, $\bR = (0,-1,-1, 
  \ldots,-1)$, $k =|\sigma|$.
\item \While{$k>1$}{
\begin{enumerate}
\item[(i)] For $1\leq i \leq k$, a party $j\in \sigma_i$ computes
  $\bPP{\bx_{\sigma_i}}$ and broadcasts it.\\ Each party computes
  $H_{\sigma_i} = H\left(\bPP{\bx_{\sigma_i}}\right)$, $1\leq i\leq
  k$.

\item[(ii)] Let $\bH$ be the sorted version of $(H_{\sigma_i} : 1\leq
  i \leq k)$, $i.e.$, assume $H_{\sigma_1}\geq H_{\sigma_2}
  \geq\cdots\geq H_{\sigma_k}$. \\ Call $\omn_{\tt id}(\sigma, \bH,
  \bR)$.  \\ Let $(\bRout, \cO)$ be its output.

\item[(iii)] Let $\sigout = \{\sigma_i:\sigma_i \in \sigma \text{
  s.t. } \sigma_i \not\subset A~\forall A \in \cO\}\bigcup\{A: A\in
  \cO\}$.  \\ Update $\bR= \bRout$, $\sigma=\sigout$, and
  $k=|\sigout|$.
\end{enumerate}
}
\end{enumerate}
\label{p:main_id}
\end{protocol}

%%%
%% \begin{remark}
%% In some practical scenarios, it is needed to solve an optimization
%% problem without knowing constraints, and such a problem has been
%% studied in the literature \cite{PapYan93, BeiCheZha15}.  With an ideal
%% decoder at our disposal, we can interpret the omniscience problem as
%% finding the optimal solution to the linear program
%% \begin{align*}
%% &\max \sum_{i=1}^mR_i \\ &R_1, \ldots, R_m \text{ s.t.}  \\ &R_A\geq
%%   H(X_A|X_{A^c}), \quad \forall\,A\subsetneq \cM,
%% \end{align*}
%% when the polymatroid imposing linear constraints is unknown but we
%% have access to an oracle (our ideal decoder) which informs us whenever
%% for a subset $A$,
%% \[
%% R_B\geq H(X_B|X_{A\setminus B}), \quad \forall\,B\subsetneq A,
%% \]
%% and, furthermore, reveals $H(X_A)$. Our protocol monotonically
%% increases the sum-rate to the optimal value for this linear program.
%% \end{remark}
%%
\begin{remark}
Recently, it was shown in \cite{ChanB15} that if a set $A$ corresponds to a
part in the partition that attains the maximum in
\eqref{e:rco_formula}, then omniscience can be attained in such a
manner that the parties in $A$ can attain omniscience along the way
from the communication of the parties in $A$.  $\rde$ explicitly
has this feature and attains omniscience for each part of the
maximizing partition along the way.
\end{remark}

We conclude this section with a few illustrative examples to
demonstrate the working of the ideal version $\rde_{\tt id}$. The first example is
for $m=3$ and exhibits a case where $\sigma_f(\cM)$ is the FDP.
\begin{example}\label{ex:1}
Let $X_1 \sim \mathsf{Ber}(1/2)$, $X_3 \sim \mathsf{Ber}(q)$, and $X_2
= X_1 \oplus X_3$.  In this case, $\cR_{\tt CO}(\{1,2,3\})$ is given
by rate vectors satisfying the following linear constraints:
\begin{align*}
R_1+R_2 &\ge 1, \\ R_2+R_3 &\ge h(q), \\ R_1+R_3 &\ge h(q).
\end{align*}
When $\frac{1}{2} < h(q) \le 1$, the finest partition is the FDP, and
\begin{align*}
R_{\tt CO}(\{1,2,3\}) = \mH_{\{1|2|3\}} = \frac{1+2h(q)}{2}.
\end{align*}
The CO region is depicted in Figure~\ref{f:region}.  As can be seen
from the figure, $R_{\tt CO}(\{1,2,3\})$ is achieved by the unique
rate assignment $\bR^* = (1/2,1/2, (2h(q)-1)/2)$. In $\rde_{\tt id}$,
parties 1 and 2 communicate first and increase their rates at slope
$1$ until $R_1= R_2 = H(X_1)-H(X_3) = H(X_2) - H(X_3) = 1-h(q)$. At
this point, party 3 starts communicating and all the parties increase
their rates at slope $1$. Owing to the initial lead of $R_1$ and $R_2$
over $R_3$, all the parties reach $\bR^*$ simultaneously.
\end{example}

\begin{figure}[t]
\begin{center}
\includegraphics[scale=0.28]{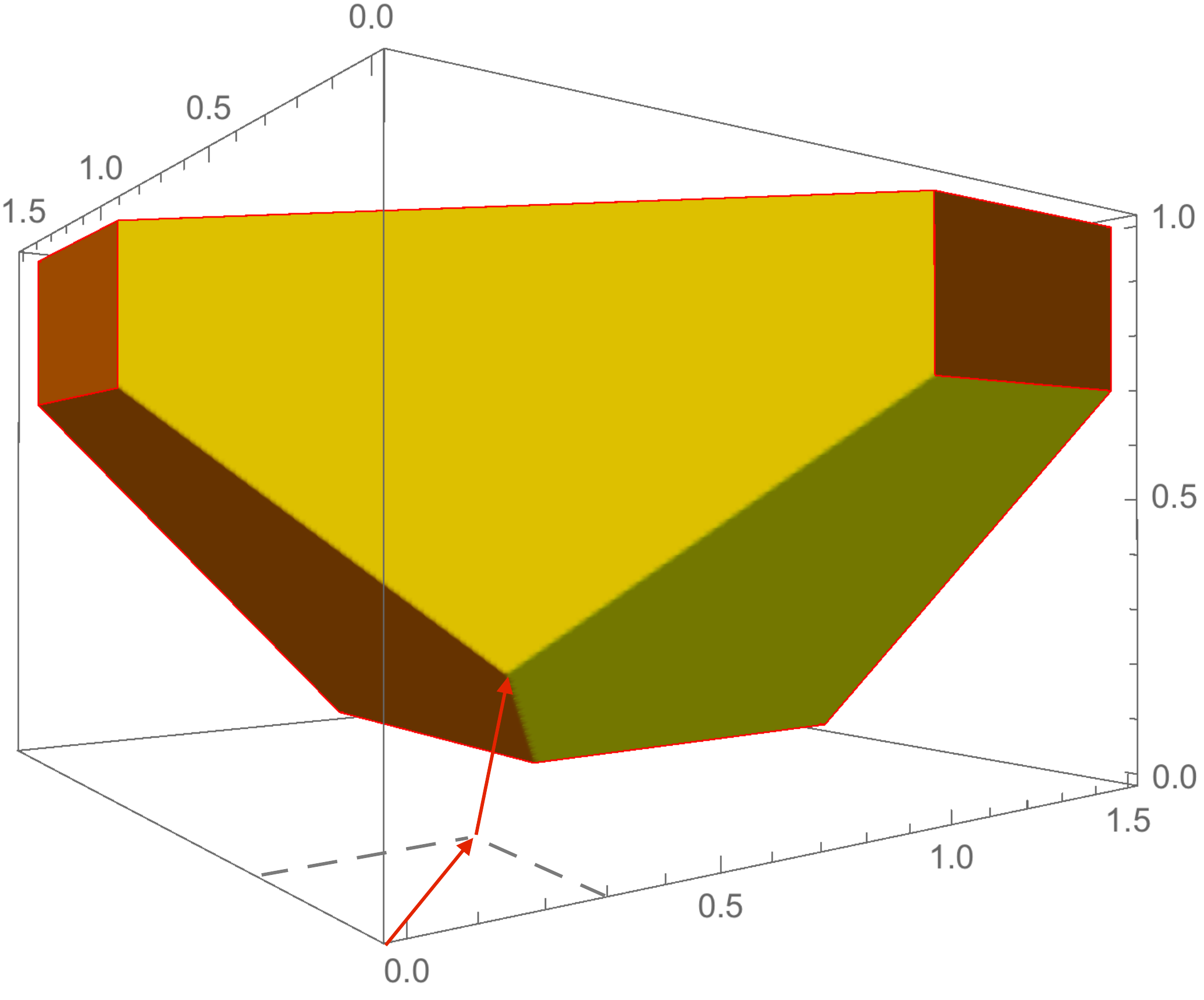}
\caption{Illustration of $\cR_{\tt CO}(\{1,2,3\})$ for
  Example~\ref{ex:1}.}
\label{f:region}
\end{center}
\end{figure}

When $\mH_\sigma$ is maximized by a partition $\sigma$ other than the
finest partition $\sigma_f(\cM)$, as $\rde_{\tt id}$ proceeds, the
parties in parts of $\sigma$ attain local omniscience, along the way,
before all the parties attain omniscience. Consider the following
example, again for $m=3$.
\begin{example}\label{ex:2}
Let $W_1,W_2 \sim \mathsf{Ber}(1/2)$ and $V_1,V_2 \sim
\mathsf{Ber}(q)$ for some $0 < q < \frac{1}{2}$, and let $X_1 =
(W_1,W_2)$, $X_2 = (W_1 \oplus V_1,W_2)$, and $X_3 = W_2 \oplus
V_2$. In this case, the partition $\{12|3\}$ is the FDP, $\mathbb{H}_{\{12|3\}} = 1 + 3h(q)$, and
$\rde_{\tt id}$ proceeds as follows: Parties 1 and 2 start increase their
rates at slope $1$. When their rates reach $h(q)$, they attain local
omniscience. At this point they start increasing their rates at slope
$1/2$ and continue doing so until $R_1+R_2$ reaches $H(X_1,X_2) -
H(X_3) = 1 + h(q)$. Now, party 3 starts communicating at slope
$1$. When all the parties reach $((1+2h(q))/2, (1+2h(q))/2,h(q))$,
they attain omniscience.
\end{example}
Note that $\{1,2\}$ attain local omniscience even before $3$ starts
communicating, illustrating the recursive structure of $\rde_{\tt id}$
wherein a subset attaining local omniscience start behaving as if the
parties in it were collocated to begin with. In fact, this recursive
property holds even when only a subset of communicating parties
attains omniscience, as our final example with $m=4$
illustrates. {The situation for $m=4$ captures the typical
  case for our general analysis -- establishing the recursive nature
  of the protocol at situations similar to that illustrated by the
  point $t_3$ in Figure~\ref{f:four-teminal-example} constitutes the
  main step in our analysis.}

\begin{example}\label{ex:3}
Let $W_1,W_2,W_3 \sim \mathsf{Ber}(1/2)$ and $V_1,V_2 \sim
\mathsf{Ber}(q)$ for some $0 < q < \frac{1}{2}$, and let $X_1 =
(W_1,W_2)$, $X_2 = (W_1 \oplus V_1,W_2)$, $X_3 = W_2 \oplus V_2$, and
$X_4 = W_3$.  Note that the observations of subset $\{1,2,3\}$ are
exactly as in Example~\ref{ex:2}. In this case, the partition $\{123|4\}$ is the FDP, $\mathbb{H}_{\{123|4\}} = 3 + 2h(q)$, and $\rde_{\tt id}$ proceeds 
as in Figure~\ref{f:four-teminal-example}.  At $t_1$, parties $1$ and
$2$ attain local omniscience and change the slopes of $R_1$ and $R_2$
to $1/2$.  At $t_2$, parties 3 and 4 start communicating. At $t_3$,
parties in $\{1,2,3\}$ attain local omniscience and change their slope
to $1/3$.  Note that up to $t_3$ the evolution of $(R_1,R_2,R_3)$ is
exactly the same as that in Example~\ref{ex:2}. Also, at $t_3$ the
rate difference $(R_1 + R_2 + R_3 - R_4)$ equals $H(X_1,X_2,X_3) -
H(X_4) = 1 + 2h(q)$. Thus, after $t_3$ the rate pair
$(R_1+R_2+R_3,R_4)$ behaves as if the parties in $\{1,2,3\}$ were
collocated to begin with. Finally, all parties attain omniscience at
$t_4$.
\end{example}

\begin{figure}[t]
\begin{center}
\includegraphics[scale=0.6]{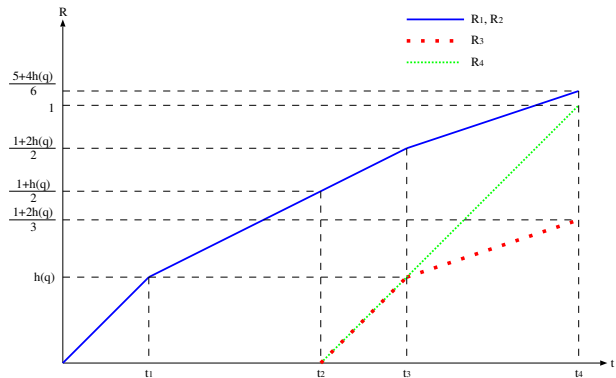}
\caption{The evolution of rates for Example~\ref{ex:3}.}
\label{f:four-teminal-example}
\end{center}
\end{figure}

%%%%%%%%%%%%%%%%%%%%%%%%%%%%%%%%%%%%%%%%%%%%%%%%%%%%%%%%%%%%%%%
\section{Universal protocol for omniscience: Full description}
Moving now to the real world, rates must be increased in discrete
increments and a positive decoding error probability must be
tolerated. To that end, the parties incrementally transmit independent
hash bits, $n\Delta$ at a time. The ideal decoder of the previous
section is replaced with a typical decoder $\dec(j, \sigma, \bR)$
which searches for the maximal set $A$ such that there exists a unique
sequence $\bx_A$ that contains the current rate vector in its CO
region and is consistent with the local observation and the received
hash values. In fact, instead of working with the original CO region
$\Rco{A}$, we use the more restrictive region $\Rcod{A}$ consisting of
vectors $(R_i, i\in \cA)$ such that
\begin{align*}
R_B \geq H(X_B\mid X_{A\setminus B}) + |B|\Delta, \quad \forall\,
B\subsetneq A.
\end{align*}
The complete decoder is described in Protocol~\ref{p:dec}.
%%%%%%%%%%%%%%%%%
\begin{protocol}[h]
\caption{$\dec(j, \sigma, \bR)$} \KwIn{An index $1\leq j\leq m$, a
  partition $\sigma\in \Sigma(\cM)$, a rate vector $\bR=(R_1, \ldots,
  R_m)$} \KwOut{A $\nack$ message, an $\ack$ message $(\ack, A)$, or
  an error message $\err$.}
\begin{enumerate}
\item {For $\sigma_i$ such that $j\in\sigma_i$, find the
  maximal set $A\subseteq \cM$ such that $\sigma_i\subsetneq A$ and
  there exists a unique sequence $\hat{\bx}_A$ such that the hashes of
  $\hat{\bx}_A$ match all the previously received hashes from parties
  in $A\setminus\{j\}$ and the joint type $\bPP{{\overline X}_A}$ of
  $\hat{\bx}_A$ satisfies the following:}
\begin{itemize}
\item[(i)] $\bPP{{\overline X}_j} = \bPP{\bx_j}$, and

\item[(ii)] $(R_i : i \in A) \in \Rcod{A\mid \bPP{{\overline X}_A}}$.
\end{itemize}
\item \uIf{there is a unique maximal $A$ found in Step 1}{return
  $(\ack,A)$.}  \ElseIf{there is no sequence found in Step 1 for any
  set $A$}{return $\nack$.}  \ElseIf{{there are multiple
    $A$s found or multiple sequences $\hat{\bx}_A$ are found for any $A$
    in Step 1}} {return $\err$.}
\end{enumerate}
\label{p:dec}
\end{protocol}
%%%%%%%%%%%%%%%%%

Note that the decoder declares $(\ack, A)$ if it can find a unique
maximal set $A$ and a unique sequence $\bx_A$, declares $\nack$ if it
finds no such set, or an $\err$ otherwise. In fact, an error may occur
even when it is not detected, $i.e.$, when $\err$ is not transmitted.
However, we can identify an event $\cE$ (described formally in
Section~\ref{s:proof_recursion}) of small probability such that under
$\cE^c$ the real decoder $\dec$ behaves exactly like $\dec_{\tt id}$,
but with $\Rco{\cA}$ replaced with $\Rcod{\cA}$. Therefore,
omniscience can be achieved in a similar manner as the ideal protocol
of the previous section.

The main component of $\rde$ is the one step
omniscience protocol $\omn$ described in Protocol~\ref{p:omn}, which
uses $\dec$ for decoding. Protocol $\omn$ proceeds very much like the
ideal protocol except that a new party $i$ starts communicating when
$R_1\geq H(\bPP{\bx_1})- H(\bPP{\bx_i}) + \alpha\Delta$, where
$\alpha\in \mN$ is an increasing threshold parameter which is updated
as the protocol proceeds. Throughout the protocol, a rate $R_i=-1$
indicates that the $i$th party is not yet transmitting and only
parties with $R_i\geq 0$ communicate. The decoder tries to attain
omniscience only among the communicating parties.
%%%%%
\begin{protocol}[t]
\caption{$\omn(\sigma, \alpha, \bH,\bR)$} \KwIn{A partition $\sigma\in
  \Sigma(\cM)$ with $|\sigma|=k$, an $\alpha\in \mN$, an entropy
  estimate vector $\bH = \left(H_{\sigma_i} : 1\leq i\leq k\right)$, a
  rate vector $\bR=(R_1, \ldots, R_m)$; we assume that $\bH$ is
  sorted, $i.e.$, $H_{\sigma_1}\geq H_{\sigma_2}\geq \cdots \geq
  H_{\sigma_k}$} \KwOut{A rate vector $\bRout$, a family of subsets
  $\cO$ that have attained omniscience.}

\begin{enumerate}
\item Initialize {$s:= \max\{i: R_{\sigma_i} \geq 0\}$}.
\item All parties $j$ such that $j\in \sigma_i$ for some
  {$1 \leq i\leq s$} send $\lceil n\Delta/|\sigma_i|\rceil$
  random hash bits. \\ Update $R_j\rightarrow R_j +
  \Delta/|\sigma_i|$. \\

\item \uIf{There exists $i>s$ such that $R_{\sigma_1}\geq
  H_{\sigma_1}- H_{\sigma_i} + \alpha\Delta$}{ set $R_j = 0$ for all
  $j\in \sigma_i$, and set {$s= \max\{i: R_{\sigma_i}\geq
    0\}$}.  }
\item For all $j$ such that $j\in \sigma_i$ for some $1\leq i \leq s$,
  execute $\dec(j, \sigma, \bR)$, which outputs $\nack$, $(\ack,
  A_j)$, or $\err$.
\item \uIf{All parties send a $\nack$} { return to Step 2.  }
  \uElseIf{No party declares an $\err$ and some parties declare an
    $\ack$,}{Identify the omniscience family
\[
\cO = \{B\subset\cM: \text{ all } j\in B \text{ returned } (\ack,
B)\}.
\]
\uIf{$\cO$ is nonempty}{Set $\bRout = \bR$, and return $(\bR,
  \cO)$.}\Else{declare an error.}} \Else{declare an error.}
\end{enumerate}
\label{p:omn}
\end{protocol}
%%%%%%%%%%%%%%%%%

The ideal protocol of the previous section works due to its recursive
structure whereby when a subset $A$ attains local omniscience, the
rate vector appears as if the parties in $A$ have been collocated from
the start. Moreover, the first subset to attain local omniscience does
so by using a communication of rate $\mH_{\sigma_f(A)}$. Both these
properties were captured by Theorem~\ref{t:recursion_id}. The result
below establishes a similar recursive property of $\omn$. However, the
definition of ``validity'' needs to be modified from
Definition~\ref{d:valid_id} -- in place of the operational definition
in the ideal case, we use the more technical definition below which
captures all the key features that we need.
%%%%
\begin{definition}\label{d:valid}
For $\alpha\in \mN$, $\sigma\in \Sigma(\cM)$ with $|\sigma|=k$ and
$\bH = (\Hsig 1, \ldots, \Hsig k)$, a rate vector $(R_1, \ldots, R_m)$
is $(\sigma, \bH, \alpha)$-valid if, for $s= \max\{i: \Rsig i \geq
0\}$, the following conditions hold:
\begin{enumerate}[(i)]
\item {\it ({Approximate constant difference})} \label{condition:constant-difference} For $1\leq i,j
  \leq s$,
\begin{align}
\Rsig i - \Rsig j \leq \Hsig i - \Hsig j + \alpha\Delta; \nonumber
\end{align}

\item {\it (Noncommunicating
  parties)} \label{condition:noncommunicating-parties}
\begin{align}
\Rsig 1 < \Hsig 1 - \Hsig {s+1} + \alpha\Delta;
\end{align}

\item {\it (Combined parties)} \label{condition:combined-parties}
  $\forall\,1\leq i \leq k$ such that $|\sigma_i|\geq 2$,
\begin{align}
(R_j : j \in \sigma_i) \in \Rcod{\sigma_i};
\label{e:part_omn_condition}
\end{align}

\item {\it (Separate parts)} \label{condition:separate-parts} for all
  $A\subseteq \{1,\ldots, k\}$ with $|A|\geq 2$,
\begin{align}
(R_j: j \in \sigma_i, i \in A) \notin \Rcod {\bigcup_{i \in
      A}\sigma_i}.  \nonumber
\end{align}
\end{enumerate}
\end{definition}

{The constant difference condition is crucial for ensuring
  the recursive nature of $\rde$ under ideal conditions. In
  general, since the rates must be incremented in discrete steps, the
  approximate version in Condition
  (\ref{condition:constant-difference}) has been introduced in the
  place of the original constant difference condition.  For
  noncommunicating parties, Condition
  (\ref{condition:noncommunicating-parties}) must be satisfied so that
  Condition (\ref{condition:constant-difference}) is maintained for
  those parties in future rounds when they start
  communicating. Condition (\ref{condition:combined-parties}) ensures
  that the current rates are enough for parties in each part to attain
  local omniscience, while Condition (\ref{condition:separate-parts})
  ensures that $\sigma$ is the maximal partition such that the parties
  in each part can attain local omniscience at current rates.}
  
%%%%%%%%%%%%%%%%%%%%%
The following theorem captures our key observation about $\omn$; its
proof is given in Section~\ref{s:proof_recursion}.
%%%%%%%%%%%%%%%%%%%%%%%%%
\begin{theorem}\label{t:recursion}
For $\alpha\in \mN$, $\sigma\in \Sigma(\cM)$ with $|\sigma|=k$ and
$\bH = (\Hsig 1, \ldots, \Hsig k)$ with $\Hsig 1 \geq \Hsig 2 \geq
\cdots \Hsig k$, let $\bRin = (\Rin 1, \ldots,\Rin m)$ be
$(\sigma,\bH, \alpha)$-valid. Then, if $\omn(\sigma,\alpha,
\bH,\bRin)$ is executed and error $\cE$ (defined in
Section~\ref{s:proof_recursion}) does not occur, the final rates
$\bRout$ and the omniscience family $\cO$ satisfy the following:
\begin{enumerate}[(I)]
\item For every $A \in \cO$, it holds that
\begin{enumerate}
\item \label{item:decomposition-of-A} $A$ consists of parts of $\sigma$,
  $i.e.$,
\begin{align*} 
A = \bigcup_{l=1}^c \sigma_{i_l}
\end{align*}
for some $\{ i_1,\ldots, i_c\}$, and

\item \label{item:range-of-Rout} denoting by $A_\sigma$ the set
  $\{\sigma_{i_1},\ldots,\sigma_{i_c} \}$, we have
\[
R^*_{\sigma_{i_l}}(A_\sigma) - 2 \alpha\Delta \le
R^{\mathtt{out}}_{\sigma_{i_l}} \le R^*_{\sigma_{i_l}}(A_\sigma) +
(m+2\alpha) \Delta,~1\le l \le c.
\]
\end{enumerate}

\item \label{item:validity-of-Rout} Let $\sigout\in\Sigma(\cM)$ be the
  partition obtained by combining the parts in $\sigma$ that belong to
  the same $A$ in $\cO$.  Let $H_{\sigout_i}$ denote the entropy of
  the type of $\bx_{\sigout_i}$. Then, with $\bH^{\tt out}=
  \left(H_{\sigout_i}, 1\leq i \leq |\sigout|\right)$, $\bRout$ is
  $(\sigout, \bH^{\tt out}, c^\prime_m\alpha)$-valid, where
  $c^\prime_m$ is a constant depending only on $m$.
\end{enumerate}
\end{theorem}
%%%%

We are now in a position to describe $\rde$. We begin by calling $\omn$ with $\sigma=\sigma_f(\cM)$,
$\alpha=1$, the sorted entropy estimates $\bH$ computed from marginal
empirical distributions $\bPP{\bx_i}$, and the rate vector $\bR =
(0,-1, \ldots, -1)$ indicating that party $1$ starts communicating and
every one else remains quiet. Note that $\bR$ is $(\sigma, \bH,
1)$-valid. A new party $i$ starts communicating when $R_1 \geq H_1 -
H_i + \Delta$. If no error occurs, $\omn$ will terminate when a subset
$A$ attains omniscience. In view of Theorem~\ref{t:recursion}, at this
point $R_A$ should be close to
$\mH_{\sigma_f(A)}\left(A|\bPP{\bx_A}\right)$ and the rates will be
$(\sigout, \bH^{\tt out}, c^\prime_m\alpha)$-valid. Thus, we are in a
similar situation as the first call to $\omn$ except that $\alpha$
must be replaced by $c^\prime_m\alpha$ and the parties in a single
part of $\sigout$ are behaving as if they are collocated. The protocol
proceeds by calling $\omn$ again with these updated parameters. Note
that under $\cE^c$, any party $j \in A$ for $A\in \cO$ can correctly
compute $\bPP{\bx_A}$ and transmit it using $O(\log n)$
bits. Proceeding recursively in this manner, the protocol stops when
parties in $\cM$ attain omniscience, which by
Theorem~\ref{t:recursion} can only happen when the sum-rate $R_\cM$ is
close to $\mH_{\sigma}(\cM|\bPP{\bx_\cM})$ for some partition $\sigma$
of $\cM$. Thus, omniscience will be attained in communication of rate
roughly less than $R_{\tt CO}\left(\cM|\bPP{\bx_\cM}\right)$. We
formally describe $\rde$ in Protocol~\ref{p:main} and
summarise its performance in Theorem
\ref{theorem:for-individual-sequence}.

%%%%%%%%%%%%%%%%%
\begin{protocol}[h]
\caption{$\rde$: The recursive data exchange protocol}
\begin{enumerate}
\item Initialize $\sigma= \sigma_f(\cM)$, $\bR = (0,-1,-1,
  \ldots,-1)$, $k =|\sigma|$, $\alpha=1$.
\item \While{$k>1$}{
\begin{enumerate}
\item[(i)] For $1\leq i \leq k$, a party $j\in \sigma_i$ computes
  $\bPP{\bx_{\sigma_i}}$ and broadcasts it.\\ Each party computes
  $H_{\sigma_i} = H\left(\bPP{\bx_{\sigma_i}}\right)$, $1\leq i\leq
  k$.

\item[(ii)] Let $\bH$ be the sorted version of $(H_{\sigma_i} : 1\leq
  i \leq k)$, $i.e.$, assume $H_{\sigma_1}\geq H_{\sigma_2} \geq
  \cdots \geq H_{\sigma_k}$. \\ Call $\omn(\sigma, \alpha, \bH, \bR)$.
  \\ \uIf{There is no error declared}{let $(\bRout, \cO)$ be its
    output.  }\Else{Terminate.}

\item[(iii)] Let $\sigout = \{\sigma_i:\sigma_i\in \sigma \text{
  s.t. } \sigma_i \not\subset A ~ \forall A \in \cO\}\bigcup\{A: A\in
  \cO\}$.  \\ Update $\bR= \bRout$, $\sigma=\sigout$, $k=|\sigout|$,
  and $\alpha \rightarrow c^\prime_m\alpha$.
\end{enumerate}
}
\end{enumerate}
\label{p:main}
\end{protocol}
We close with the following result claiming the universal rate
optimality of $\rde$ for every IID distribution. Proof is a
simple consequence of Theorem~\ref{t:recursion} and is given in
Section~\ref{s:technical}. Note that while Protocol~\ref{p:main} is a
variable length protocol, its fixed length variant can be obtained
simply by aborting the protocol once the total number of bits
communicated crosses $nR$.

\begin{theorem} \label{theorem:for-individual-sequence}
There exist constants $C_i > 0, i=1, \ldots, 4$ depending only on $m$
and a polynomial $p(n)$ depending on $\cX_i$, $i \in \cM$, such that
for every $\Delta > 0$ and every sequence $\mathbf{x}_\cM$, the
probability of error for Protocol~\ref{p:main} is bounded above by
\[
C_1 \left(\frac{ \log |\cX_\cM|}{\Delta}+m\right) p(n) 2^{-n\Delta}.
\]
Furthermore, if an error does not occur, the number of bits
communicated by the protocol for input $\mathbf{x}_\cM$ is bounded
above by
\begin{align} \label{eq:upper-bound-protocol-length}
n R_{\mathtt{CO}}(\cM|\bPP{\mathbf{x}_\cM}) + n C_2 \Delta +
C_3\left(\frac{\log |\cX_\cM|}{\Delta}+m\right) + C_4 \log n.
\end{align}
\end{theorem}

\begin{corollary}\label{col:main} 
For $\Delta = \frac{1}{\sqrt{n}}$ and every distribution
$\bPP{X_\cM}$, Protocol~\ref{p:main} has a probability of error
$\ep_n$ vanishing to $0$ as $n \rightarrow \infty$ and average length
$|\pi|_{\tt av}$ less than\footnote{The constant implied by
  $\cO(\sqrt{n \log n})$ depends on $\bPP{X_\cM}$;
  see~\eqref{eq:taylor} below.}
\begin{align*}
nR_{\tt CO}(\cM | \bPP{X_\cM}) + \cO( \sqrt{n \log n}).
\end{align*}
Furthermore, for a fixed $R>0$, the fixed-length variant of
Protocol~\ref{p:main} has probability of error $\ep_n$ vanishing to
$0$ as $n \rightarrow \infty$ for all distributions $\bPP{X_\cM}$ that
satisfy
\[
R> R_{\tt CO}\left(\cM|\bPP{X_\cM}\right) + \cO\left(\sqrt{n^{-1}\log
  n}\right).
\]
\end{corollary}

%The corollary follows from
%Theorem~\ref{theorem:for-individual-sequence} upon noting the
%continuity of $R_{\tt CO}(\cM\bPP{X_\cM})$ in $\bPP{X_\cM}$
%since $\bPP{\bx_\cM}\rightarrow \bPP{X_\cM}$ almost surely in the limit as $n \rightarrow
%\infty$. The claim for fixed length codes follows similarly using the convergence
% $\bPP{\bx_\cM}\rightarrow \bPP{X_\cM}$ in probability.

%\begin{remark} \label{remark:detaled-residual}
%More precisely, with some constant $C_1, C_2,C_3$ depending only on $m$, 
%the upper bounds on the error and the length of the protocol in Theorem \ref{theorem:for-individual-sequence} is given by
 %$\exp\{ - n C_1 \Delta \}$ and
%$n R_{\mathtt{CO}}(\cM|\bPP{\mathbf{x}_\cM}) + n C_2 \Delta + \frac{C_3}{\Delta} + \cO(\log n)$, respectively, where $\Delta$ may depend on $n$.
%\end{remark}

%%%%%%%%%%%%%%%%%%%%%%%%%%%%%%%%%%%%%%%%%%%%%%%
\section{Universal secret key agreement}\label{s:problem_description_SK}

Closely related to the omniscience problem is the SK agreement problem
where the parties seek to generate shared random bits which are almost
independent of the communication used to generate them. Specifically,
an $(\ep, \delta)$-SK agreement protocol consists of an interactive
communication protocol $\pi$ with public randomness $U$, private
randomness $U_i$ at Party $i$, and with the output of the $i$th party
$K_i=K_i(X_i^n, U_i, U,\Pi)$ such that there exists a $\cK$-valued
random variable $K$ satisfying the {\it recoverability condition}
\begin{align}
\bPr{K_i=K,\, \forall\, i \in \cM} \geq 1- \ep, \nonumber
\end{align}
and the {\it secrecy condition}\footnote{We assume that the public
  randomness $U$ is available to the eavesdropper.}
\begin{align}
\ttlvrn{\bPP{K\Pi U}}{\bPP{\tt unif}\times \bPP{\Pi U}}\leq \delta,
\nonumber
\end{align}
where $\bPP{\tt unif}$ denotes the uniform distribution on $\cK$.
\begin{definition}[Secret key capacity] For $\ep, \delta\in [0,1)$, a
  rate $R\geq 0$ is an $(\ep, \delta)$-achievable SK rate if there
  exists a $\cK^{(n)}$-valued $(\ep, \delta)$-SK with $\log
  |\cK^{(n)}|\geq nR$ for all $n$ sufficiently large. The supremum
  over all $(\ep,\delta)$-achievable SK rates is called the $(\ep,
  \delta)$-SK capacity, denoted $C_{\ep,
    \delta}(\cM|\bPP{X_\cM})$. The SK capacity for $\bPP{X_\cM}$ is
  given by
\[
C(\cM| \bPP{X_\cM}) = \lim_{\ep+\delta\rightarrow 0} C_{\ep,
  \delta}(\cM|\bPP{X_\cM}).
\]
\end{definition}
%% We shall omit the dependence on $\bPP{X_\cM}$ from the notation when
%% it is clear from the context.
\begin{theorem}[\cite{CsiNar04}] \label{theorem:key-capacity}
Given a distribution $\bPP{X_\cM}$,
\[
C\left(\cM|\bPP{X_\cM}\right) = H\left(X_\cM\right) - R_{\tt
  CO}\left(\cM|\bPP{X_\cM}\right).
\]
\end{theorem}
In fact, it was shown in \cite{TyaWat14, TyaWat14ii} that a strong
converse holds and $C_{\ep, \delta}(\cM|\bPP{X_\cM}) =
C(\cM|\bPP{X_\cM})$ for all $\ep+\delta<1$.

The achievability of rate $H\left(X_\cM\right) - R_{\tt
  CO}\left(\cM|\bPP{X_\cM}\right)$ was shown in \cite{CsiNar04} by
establishing a connection between SK agreement and omniscience. In
particular, a SK achieving capacity was generated by first
communicating at rate $R_{\tt CO}\left(\cM|\bPP{X_\cM}\right)$ to
attain omniscience, and then extracting a SK from $X_\cM^n$ which is
almost independent of the communication used for
omniscience. Following the same methodology, we provide a universal SK
agreement protocol which builds upon the universal omniscience
protocol of the previous section.

We consider a slight generalization of the definition of SK above,
which admits variable length SKs. An $(\ep, \delta)$-SK $K$ and its
estimates $K_1, \ldots, K_m$ now take values in $\cK = \{0,1\}^*$, the
set of finite length binary sequences. The recoverability condition
remains as before. However, the secrecy condition needs to be
modified. Specifically, denoting by $T$ the random length of $K$,
which we assume to be available to the eavesdropper, the secrecy
condition now requires
\begin{align}
\sum_{t} \bP T t\ttlvrn{\bPP{K\Pi U|T = t}}{\bPP{{\tt unif}, t}\times
  \bPP{\Pi U| T=t}}\leq \delta, \nonumber
\end{align}
where $\bPP{{\tt unif}, t}$ denotes the uniform distribution on
$\{0,1\}^t$. The average achievable rate and average SK capacity are
defined as above with the worst-case length $\log |\cK|$ replaced by
the average length $\bEE[T]$. Instead of introducing a new notation
for average SK capacity, we note that it equals $C\left(\cM |
\bPP{X_\cM}\right)$ and, with an abuse of notation, use $C\left(\cM |
\bPP{X_\cM}\right)$ to denote both the SK capacity and the average SK
capacity. Indeed, the achievability is the same as above since a fixed
length SK constitutes a variable length SK. For the converse, denoting
\begin{align*}
\ep_t := 1-\bPr{K = K_i, i \in \cM| T = t} \text{ and } \delta_t :=
\ttlvrn{\bPP{K\Pi U|T = t}}{\bPP{{\tt unif}, t}\times \bPP{\Pi U|
    T=t}},
\end{align*}
it follows by applying the converse proof of \cite{CsiNar04} for each
fixed value $T=t$ that
\[
\frac t n \leq C\left(\cM | \bPP{X_\cM}\right) + g_1(\ep_t) +
g_2(\delta_t),
\]
where $g_1$ and $g_2$ are concave, increasing functions satisfying
$g_i(x) \rightarrow 0$ as $x\rightarrow 0$. Thus,
\begin{align*}
\frac {\bEE[T]} n &\leq C\left(\cM | \bPP{X_\cM}\right) +
\bEE[g_1(\ep_T) + g_2(\delta_T)] \\ &\leq C\left(\cM |
\bPP{X_\cM}\right) + g_1(\bEE [\ep_T]) + g_2(\bEE[\delta_T]) \\ &\leq
C\left(\cM | \bPP{X_\cM}\right) + g_1(\ep) + g_2(\delta),
\end{align*}
where the last two inequalities hold since $g_i$, $i=1,2$, are concave
and increasing.

We present a universal SK agreement protocol that generates a SK of
average length $nC(\cM| \bPP{X_\cM}) - \cO(\sqrt{n \log n})$ without
the knowledge of the underlying distribution
$\bPP{X_\cM}$. Specifically, first the parties use
Protocol~\ref{p:main} with $\Delta=1/\sqrt{n}$ to recover
$X_\cM^n$. If no error occurs and the recovered sequence is $\bx_\cM$,
by Theorem~\ref{theorem:for-individual-sequence} the number of bits
communicated is no more than
\[
l(\bx_\cM) = n R_{\mathtt{CO}}(\cM|\bPP{\mathbf{x}_\cM}) +
\cO(\sqrt{n}).
\]
We extract a SK from recovered $\bx_\cM^n$ by randomly
hashing\footnote{The random hash can be replaced by a randomly
  selected member of a $2$-universal hash family.} it to roughly
$nH(\bPP{\bx_\cM})- l(\bx_\cM)$ values. Formal description of the
protocol is given in Protocol~\ref{p:SK_agreement}; the length of the
SK is tuned to the secrecy parameter $\delta$.

\begin{protocol}[h]
\caption{A universal SK agreement protocol} \KwIn{Step size parameter
  $\Delta$ and secrecy parameter $\delta$}
\begin{enumerate}
\item Parties execute Protocol~\ref{p:main} with step-size $\Delta$.
\item \uIf{Protocol~\ref{p:main} completes without declaring an error}
  {%set $l$ equal to the number of bits communicated in
    Protocol~\ref{p:main}.\\ Each party $i\in \cM$ forms an estimate
    $K_i$ of the SK as follows:

(i) Denoting by $\dP^{(i)}$ the type of the estimate $\bx_\cM^{(i)}$
    of $X_\cM^n$ at Party $i$ and by $\ep_n$ the maximum error
    probability of Protocol~\ref{p:main}, set
    $l\left(\dP^{(i)}\right)$ to be the quantity in
    \eqref{eq:upper-bound-protocol-length} for $\dP^{(i)}$ and
\[
k\big(\dP^{(i)}\big) = nH\big(\dP^{(i)}\big) - l\big(\dP^{(i)}\big) -
|\cX_\cM|\log (n+1) - 2 \log \frac 1{\delta-2\ep_n} + 2;
\]

(ii) generate $K_i$ by randomly hashing $\bx_\cM^{(i)}$ to
$k\left(\dP^{(i)}\right)$ bits.  } \Else{Declare an error.}
\end{enumerate}
\label{p:SK_agreement}
\end{protocol}

\begin{theorem}\label{t:performance_SK_protocol}
For $\Delta = \frac{1}{\sqrt{n}}$, $0<\delta<1$, and every
distribution $\bPP{X_\cM}$, Protocol~\ref{p:SK_agreement} generates a
variable length $(\ep_n, \delta)$-SK with $\ep_n$ vanishing to $0$ as
$n \rightarrow \infty$ and average length greater than
\begin{align} \label{eq:lower-bound-average-key-length}
nC(\cM | \bPP{X_\cM}) - \cO( \sqrt{n \log n}).
\end{align}
\end{theorem}

%%%%%%%%%%%%%%%%%%%%%%%%%%%%%%%%%%%%%%%%%%%%%%%
\section{Technical results and proofs}\label{s:technical}
This section contains the proofs of our results. We begin by noting a
few properties of the mathematical quantities involved in our proofs.

%%%%%%%%%%%%%%%%%%%%%
\subsection{Properties of CO region and related quantities}

With a general subset $A\subseteq \cM$ in the role of $\cM$,
 {we define the notations $\Rstar{i}(A)$ and
  $\oH_{\sigma}(A)$, $\sigma\in \Sigma(A)$ in a similar manner as in
  \eqref{e:R_star_definition} and \eqref{e:H_sigma_definition},
  respectively.}  Our first lemma notes some simple properties of
$\oH_{\sigma}(A)$ and $\Rstar{i}(A)$.
%%%%
\begin{lemma}\label{l:rstar_char}
For $A\subseteq \cM$ and $\sigma\in \Sigma(A)$, the following
relations hold between $\Rstar {i}(A)$ and $\oH_{\sigma}(A)$:
\begin{align}
\sum_{i\in A}\Rstar {i}(A)&=\oH_{\sigma_f}(A);
\label{e:rstar_sum_formula2}
\\ \Rstar {i}(A) &= \oH_{\sigma_f}(A) - H(X_A|X_i), \quad \forall i
\in A;
\label{e:rstar_formula1}
\\ \sum_{i\in B}\Rstar {i}(A) - H(X_A|X_{A\setminus B})&=
|B|\bigg[\oH_{\sigma_f}(A) - \oH_{\sigma_B}(A) \bigg],
\label{e:rstar_sum_formula1}
\end{align}
where the final equality holds for every $B \subsetneq A$, with the
shorthand $\sigma_B$ for the partition $\sigma_B(A) \in \Sigma(A)$
given by $\{\{A\setminus B\}, \{i\} : i \in B\}$.

Furthermore, $\Rstar {i}(A)$ satisfies the following properties:
\begin{align}
\sum_{j \in A}\Rstar j(A) - \Rstar i (A) &= H(X_A|X_i), \quad
\forall\, i\in A;
\label{e:rstar_sum_formula3}
\\ \Rstar i(A) - \Rstar j (A) &= H(X_i) - H(X_j), \quad \forall\, i,j
\in A.
\label{e:rstar_diff_formula}
\end{align}
Finally, for $A \subseteq \cM$ and $\sigma\in \Sigma(A)$, similar
results holds for $\Rstar{\sigma_i}(A_\sigma)$, $1\leq i \leq
|\sigma|$, with $A_\sigma$ in place of $A$.
\end{lemma}
%%%
\begin{proof}
Since $(R_i^*(A) : i \in A)$ is the solution of
\begin{align*}
\sum_{j \in A \atop j \neq i} R_j = H(X_A|X_i),~i\in A,
\end{align*}
by taking the summation of all the constraints and by dividing by
$|A|-1$, we have \eqref{e:rstar_sum_formula2}.  Then, by subtracting
the constraint for $i$ from \eqref{e:rstar_sum_formula2}, we have
\eqref{e:rstar_formula1}.  From \eqref{e:rstar_formula1}, for every $B
\subsetneq A$ it holds that
\begin{align} \label{eq:proof-e:rstar_sum_formula1}
\sum_{i \in B} R^*_i(A) = |B| \oH_{\sigma_f}(A) - \sum_{i \in B} H(X_A
| X_i).
\end{align}
Also,
\begin{align*}
\oH_{\sigma_B}(A) = \frac{1}{|B|} \left[ \sum_{i \in B} H(X_A | X_i) +
  H(X_A | X_{A\backslash B}) \right],
\end{align*}
which is equivalent to
\begin{align*}
\sum_{i \in B} H(X_A | X_i) = |B| \oH_{\sigma_B}(A) - H(X_A |
X_{A\backslash B}).
\end{align*}
Combining this with \eqref{eq:proof-e:rstar_sum_formula1}, we have
\eqref{e:rstar_sum_formula1}.

By taking the difference of \eqref{e:rstar_sum_formula2} and
\eqref{e:rstar_formula1}, we have \eqref{e:rstar_sum_formula3};
\eqref{e:rstar_diff_formula} also follows from
\eqref{e:rstar_formula1}. The final statement is proved exactly in the
same manner by regarding $X_{\sigma_i}$ as a single random variable.
\end{proof}
%%%%%%%%%%%%%%%%%%%%%
Next, we prove another useful relation between $\oH$ and $\Rstar{i}$
showing that the difference $\sum_{i\in B}\Rstar{i}(A) -
\oH_{\sigma_f}(B)$ must have the same sign as
$\oH_{\sigma_{\overline{B}}}(A) - \oH_{\sigma_f}(A)$, where $\overline
B$ denotes $A \backslash B$ and, as before, we have used the shorthand
$\sigma_B$ for the partition $\sigma_B(A)$ of $A$.
%%%%%%%%%%%%%%%%%%%%%
\begin{lemma}\label{l:rstar_alt_char}
For every $B \subsetneq A \subseteq \cM$ with $\overline{B}= A
\backslash B$,
\begin{align}
\sum_{i\in B}\Rstar{i}(A) = \oH_{\sigma_f}(B)
+\frac{|B||\overline{B}|}{|B|-1}\bigg[ \oH_{\sigma_{\overline{B}}}(A)
  - \oH_{\sigma_f}(A) \bigg].
\label{eq:rstar_alt_char}
\end{align}
For $A \subseteq \cM$ and $\sigma\in \Sigma(A)$, similar results holds
for $\Rstar{\sigma_i}(A_\sigma)$, $1\leq i \leq |\sigma|$, with $B
\subsetneq A_\sigma$ in place of $B \subsetneq A$.
\end{lemma}
%%%
\begin{proof}
First we have
\begin{align}
\lefteqn{ (|B|-1)\left[ R^*_B(A) - \oH_{\sigma_f}(B) \right] }
\nonumber \\ &= (|B|-1) R^*_B(A) - \sum_{i \in B} H(X_B|X_i) \nonumber
\\ &= (|B|-1) \left[ |B| \oH_{\sigma_f}(A) - \sum_{i \in B} H(X_A|X_i)
  \right] - \sum_{i \in B} H(X_B|X_i) \nonumber \\ &=
\frac{|B|-1}{|A|-1} \left[ |B| \sum_{i \in A} H(X_A| X_i) - (|A|-1)
  \sum_{i \in B} H(X_A|X_i) \right] - \sum_{i \in B} H(X_B|X_i)
\nonumber \\ &= \frac{(|B|-1)|B|}{|A|-1} \sum_{i \in \overline{B}}
H(X_A|X_i) + \left( \frac{(|B|-1)|B|}{|A|-1} - |B| \right) \sum_{i \in
  B} H(X_B|X_i) \nonumber \\ &~~~ + \left( \frac{(|B|-1)|B|}{|A|-1} -
(|B|-1) \right) |B| H(X_A | X_B), \label{eq:proof-eq:rstar_alt_char-1}
\end{align}
where we used \eqref{e:rstar_formula1} in the second equality. On the
other hand, we have
\begin{align}
\lefteqn{ |B| |\overline{B}| \left[ \oH_{\sigma_{\overline{B}}}(A) -
    \oH_{\sigma_f}(A) \right] } \nonumber \\ &= \frac{|B|}{|A|-1}
\left[ (|A|-1) \sum_{i \in \overline{B}} H(X_A|X_i) + (|A|-1) H(X_A |
  X_B) - |\overline{B}| \sum_{i\in A} H(X_A|X_i) \right] \nonumber
\\ &= \frac{|B|}{|A|-1} \left[ (|B|-1) \sum_{i \in \overline{B}} H(X_A
  | X_i) + (|A|-1) H(X_A|X_B) - |\overline{B}| \sum_{i\in B} \left(
  H(X_B|X_i) + H(X_A|X_B) \right) \right] \nonumber \\ &=
\frac{(|B|-1) |B|}{|A|-1} \sum_{i \in \overline{B}} H(X_A|X_i) -
\frac{|B| |\overline{B}|}{|A|-1} \sum_{i \in B} H(X_B|X_i) \nonumber
\\ &~~~+ \frac{|B|}{|A|-1} \left( (|A|-1) - |B||\overline{B}| \right)
H(X_A|X_B), \label{eq:proof-eq:rstar_alt_char-2}
\end{align}
where we used $|\overline{B}| = |A| - |B|$ in the second equality.  We
can verify that the coefficient of each term in
\eqref{eq:proof-eq:rstar_alt_char-1} and
\eqref{eq:proof-eq:rstar_alt_char-2} coincides.  Thus, we have
\eqref{eq:rstar_alt_char}.

The second statement is proved exactly in the same manner by regarding
$X_{\sigma_i}$ as a single random variable.
\end{proof}
%%%%%%%%%%%%%%%%%%%%%
As $\rde$ proceeds, subsets of parties that have attained local
omniscience start behaving as one. In the next recursive step of the
protocol such sets of parties behaving as one attain omniscience. The
next lemma ensures that when the rate is sufficient for these sets of
parties to attain omniscience, it is sufficient also for the
individual members of these sets to attain omniscience.
%%%%%%%%%%%%%%%%%%%%%
\begin{lemma}\label{l:combine_omn}
For a subset $A \subseteq \cM$ and a partition $\sigma\in \Sigma(A)$
with $|\sigma|=k$, suppose that for every $1\leq i \leq k$
\begin{align} \label{eq:condition-combine_omn-1}
\left(R_j: j \in \sigma_i\right) \in \Rcod{\sigma_i},
\end{align}
and
\begin{align} \label{eq:condition-combine_omn-2}
\left(\Rsig i: 1\leq i \leq k\right) \in \Rcod{A_\sigma},
\end{align}
where the elements of the set $A_\sigma$ consist of parts of the
partition $\sigma$ (each part treated as a single element).  Then, it
holds that
\begin{align}
\left(R_i: i \in A\right)\in \Rcod{A}.  \nonumber
\end{align}
\end{lemma}
%%%%%
\begin{proof}
We prove that for any $B \subsetneq A$,
\begin{align*}
R_B \ge H(X_B | X_{A\backslash B}) + |B|\Delta.
\end{align*}
Without loss of generality, we can assume
\begin{align*}
B = \left( \bigcup_{i=1}^{k^\prime} B_i \right) \cup \left(
\bigcup_{i=k^\prime+1}^k B_i \right)
\end{align*}
for some $1 \le k^\prime \le k$, where $B_i \subsetneq \sigma_i$ for
$1 \le i \le k^\prime$ and $B_i = \sigma_i$ for $k^\prime +1 \le i \le
k$ ($B_i$ may be empty set for $1 \le i \le k^\prime$). Then, from
\eqref{eq:condition-combine_omn-1} and
\eqref{eq:condition-combine_omn-2}, we have
\begin{align*}
R_B &= \sum_{i =1}^{k^\prime} R_{B_i} + \sum_{i = k^\prime +1}^k
R_{B_i} \\ &\ge \sum_{i = 1}^{k^\prime} H(X_{B_i} | X_{\sigma_i
  \backslash B_i}) + H(X_{\sigma_{k^\prime+1}},\ldots,X_{\sigma_{k}} |
X_{A_\sigma \backslash \{ \sigma_{k^\prime+1},\ldots,\sigma_{k}\}}) +
|B| \Delta \\ &= \sum_{i = 1}^{k^\prime} H(X_{B_i} | X_{\sigma_i
  \backslash B_i}) + H(X_{B_{k^\prime+1}},\ldots,X_{B_{k}} | X_{A
  \backslash \cup_{i=k^\prime+1}^{k} B_i }) + |B| \Delta \\ &\ge
\sum_{i = 1}^{k^\prime} H(X_{B_i} | X_{A \backslash \cup_{j=i}^{k}
  B_j}) + H(X_{B_{k^\prime+1}},\ldots,X_{B_{k}} | X_{A \backslash
  \cup_{i=k^\prime+1}^{k} B_i }) + |B| \Delta \\ &= H(X_B |
X_{A\backslash B}) + |B| \Delta.
\end{align*}
\end{proof}
%%%%%%%%%%%%%%%%%%%%%
The next observation helps us to relax the assumptions of the previous
lemma by showing that the collocated parts of a partition will attain
local omniscience even if a collection of (nonempty) subsets of each
part attains local omniscience.
%%%%%%%%%%%%%%%%%%%%%
\begin{lemma}\label{l:completion_omn}
For a subset $A \subseteq \cM$ and a partition $\sigma\in \Sigma(A)$
with $|\sigma|=k$, let $B_i \subseteq \sigma_i$ be nonempty for $1\leq
i \leq k$. Suppose that for every $1\leq i \leq k$
\begin{align} \label{eq:condition-completion_omn-1}
\left(R_j: j \in \sigma_i\right) \in \Rcod{\sigma_i},
\end{align}
and
\begin{align} \label{eq:condition-completion_omn-2}
\left(R_j: j \in B_i, 1\leq i \leq k\right) \in \Rcod{\bigcup_{i=1}^k
  B_i}.
\end{align}
Then, it holds that
\begin{align}
\left(R_{\sigma_i}: 1\leq i \leq k\right)\in \Rcod{A_\sigma}.
\nonumber
\end{align}
\end{lemma}
%%%
\begin{proof}
It suffices to show that for any $C= \cup_{i=1}^c \sigma_i$ with
$1\leq c\leq k$
\begin{align}
R_C \ge H(X_C | X_{A_\sigma \backslash C}) + |C| \Delta.
\end{align}
To that end, we have from \eqref{eq:condition-completion_omn-1} and
\eqref{eq:condition-completion_omn-2} that
\begin{align*}
R_C &= \sum_{i =1}^c R_{\sigma_i} \\ &= \sum_{i =1}^c R_{\sigma_i
  \backslash B_i} + \sum_{i=1}^c R_{B_i} \\ &\ge \sum_{i=1}^c
H(X_{\sigma_i \backslash B_i} | X_{B_i}) + H(X_{B_1},\ldots,X_{B_c} |
X_{B_{c+1}},\ldots, X_{B_k}) + |C| \Delta \\ &\ge \sum_{i=1}^c
H(X_{\sigma_i \backslash B_i} | X_{\cup_{j=1}^{i-1} (\sigma_j
  \backslash B_j)}, X_{B_1},\ldots, X_{B_c}, X_{\sigma_{c+1}},\ldots,
X_{\sigma_k}) \\ &~~~~ + H(X_{B_1},\ldots,X_{B_c} |
X_{\sigma_{c+1}},\ldots,X_{\sigma_k}) + |C| \Delta \\ &= H(X_C |
X_{A_\sigma \backslash C}) + |C| \Delta.
\end{align*}
\end{proof}
%%%%
We need to show that when each call to $\omn$
terminates, which happens when a subset $A$ attains local omniscience,
the rate of communication used for each party $i\in A$ is
$\Rstar{i}(A)$ (or if $\omn$ was called with a partition $\sigma$ then
the same property holds with $i$ and $A$, respectively, replaced by $\sigma_l$ and
$A_\sigma$, where $A_\sigma$ is the set of parts that comprise
$A$). Recall that $\omn$ ensures that for each
communicating party (or a set consisting of collocated parties) the
difference $R_{i} - \Rstar{i}(A)$ is maintained for every $A\subseteq
\cM$. Therefore, all the parties in $A$ will reach $\Rstar{i}(A)$ at
the same time, and, since before reaching this rate their sum-rate
will not be sufficient for omniscience, it suffices to show that the
rate vector $(\Rstar{i}(A) : i \in A)$ lies the omniscience region for
$A$. The next technical lemma shows that there must be some subset $A$
for which this holds and constitutes the main step in our proof. We
show a slight generalization which holds when the parties in parts of
$\sigma\in \Sigma(\cM)$ are collocated.

%%%%
\begin{lemma} \label{lemma:induction}
For a partition $\sigma \in \Sigma(\cM)$ and $A \subseteq \cM$ such
that
\begin{align*}
A = \bigcup_{l=1}^c \sigma_{i_l},
\end{align*}
there exists $B \subseteq \{1,\ldots,c\}$ with $|B| \ge 2$ such that
\begin{align} \label{eq:induction}
( R_{\sigma_{i_l}}^*(A_\sigma) + |\sigma_{i_l}| \Delta : l \in B) \in
  \cR_{\mathtt{CO}}^\Delta( \{ \sigma_{i_l}: l \in B\}),
\end{align}
where $A_\sigma$ is the set of parts $\sigma_i$ that comprise $A$,
with each part treated as a single element.
\end{lemma}
%%%
\begin{proof}
Since \eqref{eq:induction} is equivalent to
\begin{align*}
( R_{\sigma_{i_l}}^*(A_\sigma) : l \in B) \in \cR_{\mathtt{CO}}( \{
  \sigma_{i_l}: l \in B\}),
\end{align*}
we prove the claim for $\Delta = 0$.  We proceed by induction on
$c$. For $c=2$, since
\begin{align*}
R_{\sigma_{i_1}}^*(A_\sigma) &= H(X_{\sigma_{i_1}} |
X_{\sigma_{i_2}}), \\ R_{\sigma_{i_2}}^*(A_\sigma) &=
H(X_{\sigma_{i_2}} | X_{\sigma_{i_1}}),
\end{align*}
$B = \{1,2\}$ satisfies the claim.  Suppose that the claim holds for
all $c \le b$. For $c=b+1$, if
\begin{align*}
(R_{\sigma_{i_l}}^*(A_\sigma) : 1 \le l \le c) \in
  \cR_{\mathtt{CO}}(A_\sigma),
\end{align*}
then $B = \{1,\ldots,c\}$ satisfies the claim. Otherwise, there exists
$\overline{C} \subsetneq \{1,\ldots, c\}$ with $\overline{C} \neq
\emptyset$ such that, with $C = \{1,\ldots,c\}\backslash
\overline{C}$,
\begin{align*}
\sum_{l \in \overline{C}} R_{\sigma_{i_l}}^*(A_\sigma) < H( X_{\cup_{l
    \in \overline{C}} \sigma_{i_l}} | X_{\cup_{l \in C}
  \sigma_{i_l}}).
\end{align*}
Then, by Lemma \ref{l:rstar_alt_char} and \eqref{e:rstar_sum_formula1}
of Lemma \ref{l:rstar_char}, it holds that
\begin{align}
\sum_{l \in C} R_{\sigma_{i_l}}^*(A_\sigma) &>
\mathbb{H}_{\sigma_f}(\{ \sigma_{i_j}: j \in C \}) \nonumber \\ &=
\sum_{l \in C} R_{\sigma_{i_l}}^*(\{ \sigma_{i_j} : j \in C \}).
  \label{eq:proof-induction-1}
\end{align} 
Since
\begin{align*}
R_{\sigma_{i_l}}^*(A_\sigma) - R_{\sigma_{i_{l^\prime}}}^*(A_\sigma)
&= H(X_{\sigma_{i_l}}) - H(X_{\sigma_{i_{l^\prime}}}) \\ &=
R_{\sigma_{i_l}}^*(\{ \sigma_{i_j} : j \in C \}) -
R_{\sigma_{i_{l^\prime}}}^*(\{ \sigma_{i_j} : j \in C \})
\end{align*}
for every $l \neq l^\prime$ ($cf.$~\eqref{e:rstar_diff_formula} of
Lemma \ref{l:rstar_char}), \eqref{eq:proof-induction-1} implies
\begin{align} \label{eq:proof-induction-2}
R_{\sigma_{i_l}}^*(A_\sigma) > R_{\sigma_{i_l}}^*(\{ \sigma_{i_j} : j
\in C \}),~\forall\, l \in C.
\end{align}
Since $|C| \le b$, by the induction hypothesis, there exists $B
\subseteq C$ such that
\begin{align*}
(R_{\sigma_{i_l}}^*(\{ \sigma_{i_j} : j \in C\}) : l \in B) \in
  \cR_{\mathtt{CO}}( \{ \sigma_{i_l} : l \in B\}),
\end{align*}
which together with \eqref{eq:proof-induction-2} implies that $B$
satisfies the claim for $c=b+1$.
\end{proof}
%%%%%%%%%%%%%%%%%%%%%
We are now in a position to prove the main results.

%%%%%%%%%%%%%%%%%%%%%
\subsection{Proof of Theorem~\ref{t:recursion}}\label{s:proof_recursion}

Before going to the proof of Theorem~\ref{t:recursion}, let us
formally define the error event $\cE$.  We shall show that this event
will happen with vanishing probability.

 {Let $L$ denote the maximum number of rounds of communication for party 1 over 
  all possible values $\bx_\cM$ of the data sequence\footnote{Since party 1 is the first one to communicate
    and continues to communicate till the last round in $\rde$,
    the number of times other parties communicate does not exceed
    $L$.}.}
Since the protocol terminates either correctly or erroneously once the
rate vector enters the omniscience region, and since
\begin{align*}
\big( \log |\cX_{\sigma_i}| + |\sigma_i| \Delta : 1 \le i \le |\sigma|
\big) \in \cR_{\mathtt{CO}}^\Delta(\cM_\sigma | \bPP{\mathbf{x}})
\end{align*}
holds for any $\sigma \in \Sigma(\cM)$ and any sequence $\mathbf{x}$,
$L$ is bounded above as
\begin{align}
L &\le \max_{\sigma \in \Sigma(\cM)} \max_{1\le i \le |\sigma|}
\frac{\log |\cX_{\sigma_i}|}{\Delta} + |\sigma_i| \nonumber \\ &\le
\frac{\log |\cX_\cM|}{\Delta} +m.
\label{e:L_bound}
\end{align}
For a fixed $\mathbf{x} \in \cX_\cM^n$, let $L(\mathbf{x})$ be the
maximum number of rounds of communication when $\mathbf{x}$ is
observed by the parties.  With a slight abuse of notation, denote by
$\mathbf{R}(l)$ the rate of communication after $l$ rounds, $1 \le l
\le L$, if the protocol does not declare an error till then. Also, let
$h_l(\mathbf{x}_i)$ denote the random hash bits sent by the $i$th
party (observing $\bx_i$) in the $l$th round.

For $B \subsetneq A \subseteq \cM$ and $1 \le l \le L(\mathbf{x})$,
let
\begin{align*}\displaystyle
\cT_l^B(\mathbf{x}) = \left\{ \mathbf{x}_A^\prime : (R_i(l) : i \in A)
\in \cR_{\mathtt{CO}}^\Delta(A | \bPP{\mathbf{x}_A^\prime}) \mbox{ and
} \{ i \in A : \mathbf{x}_i^\prime \neq \mathbf{x}_i \} = B \right\}.
\end{align*}
Note that
\begin{align}
|\cT_l^B(\mathbf{x})| &\le p(n) \max_{\bPP{\overline{X}_A} \in
  \cP_n(\cX_A) : \atop (R_i(l) : i \in A) \in
  \cR_{\mathtt{CO}}^\Delta\left(A|\bPP{\overline{X}_A}\right)} |\{
\mathbf{x}_A^{\prime} : \bPP{\mathbf{x}_A^{\prime}} =
\bPP{\overline{X}_A}, \{ i \in A : \mathbf{x}_i^\prime \neq
\mathbf{x}_i \} = B \}| \nonumber \\ &\le p(n)
\max_{\bPP{\overline{X}_A} \in \cP_n(\cX_A) : \atop (R_i(l) : i \in A)
  \in \cR_{\mathtt{CO}}^\Delta\left(A|\bPP{\overline{X}_A}\right)}
2^{n H\left( \overline{X}_A | \overline{X}_{A\backslash B} \right) }
\nonumber \\ &\le p(n) 2^{nR_B(l) - n
  |B|\Delta}, \label{eq:bound-on-typical-set}
\end{align}
where $\cP_n(\cX_A)$ is the set of all types on $\cX_A$ and $p(n)$ is
the number of types and is polynomial in $n$.

For $B \subsetneq A$, denote by $\cE_A(l,B)$ the error event
\begin{align*}
\cE_A(l,B) = \left\{ \exists \mathbf{x}_A^\prime \in
\cT_l^B(\mathbf{x}) \mbox{ s.t. } h_k(\mathbf{x}_j) =
h_k(\mathbf{x}_j^\prime)~\forall j \in B, \forall 1 \le k \le l
\right\}.
\end{align*}
Finally, let
\begin{align*}
\cE_A(l) &= \bigcup\left\{ \cE_A(l,B) : B \neq \emptyset, B \subsetneq
A \right\}, \\ \cE &= \bigcup \left\{ \cE_A(l) : 1 \le l \le
L(\mathbf{x}), A \subseteq \cM \right\}.
\end{align*}

%%%
\begin{lemma} \label{lemma:decoding-error}
There exists a constant $C$ depending only on $m$ such that, for every
sequence $\mathbf{x} \in \cX_\cM^n$, the probability of the error
event $\cE = \cE(\mathbf{x})$ defined above is bounded by $C L p(n)
2^{- n \Delta}$.
\end{lemma}
%%%

%%%
\begin{remark} \label{remark:error-event}
Suppose for a sequence $\mathbf{x} \in \cX_\cM^n$, error $\cE$ does
not occur. By definitions of $\cE$ and $\dec$, when $\mathrm{OMN}$
terminates, a set $A$ belongs to $\cO$ if and only if the final rates
of communication $\mathbf{R}$ satisfy $(R_i: i \in A) \in
\cR_{\mathtt{CO}}^\Delta(A|\bPP{\mathbf{x}})$.
\end{remark}
%%%

\begin{proof}
By noting the bound in \eqref{eq:bound-on-typical-set}, we have
\begin{align*}
\Pr( \cE_A(l,B) \mid X_\cM^n = \mathbf{x}) &\le
\frac{1}{2^{{nR_B(l)}}} |\cT_l^B(\mathbf{x})| \\ &\le p(n) 2^{-n
  \Delta},
\end{align*}
where the first inequality uses the bound for probability of collision event $h_k(\bx_j) = h_k(\bx_j^\prime)$ which holds for a random hash. (In fact, the same bound holds for a randomly selected member of a $2$-universal hash family.) 
Thus,
\begin{align*}
\Pr(\cE \mid X_\cM^n = \mathbf{x}) &\le 2^m \cdot L \cdot \max_{A
  \subseteq \cM, 1 \le l \le L(\mathbf{x})} \Pr( \cE_A(l) \mid X_\cM^n
= \mathbf{x}) \\ &\le 4^m \cdot L \cdot \max_{A \subseteq \cM, 1 \le l
  \le L(\mathbf{x}), B \subsetneq A} \Pr( \cE_A(l,B) \mid X_\cM^n =
\mathbf{x}) \\ &\le 4^m \cdot L \cdot p(n) \cdot 2^{-n\Delta}.
\end{align*}
\end{proof}

{\it Proof of Theorem~\ref{t:recursion}.} We prove each statement of
Theorem~\ref{t:recursion} separately.

\paragraph*{Proof of (\ref{item:decomposition-of-A})}
Denoting $B_i = A \cap \sigma_i$, let $\{i_1,\ldots,i_c\}$ be those
indices $i \in \{1,\ldots,k\}$ for which $B_i \neq \emptyset$.  Since
$\mathbf{R}^{\mathtt{in}}$ is $(\sigma,\mathbf{H},\alpha)$-valid, it
satisfies
\begin{align*}
(R_j^{\mathtt{in}}: j \in \sigma_i) \in
  \cR_{\mathtt{CO}}^\Delta(\sigma_i),~\forall\, 1 \le i \le k \mbox{
    s.t. } |\sigma_i| \ge 2,
\end{align*}
and therefore, so does $\mathbf{R}^{\mathtt{out}}$, i.e.,
\begin{align*}
(R_j^{\mathtt{out}}: j \in \sigma_i) \in
  \cR_{\mathtt{CO}}^\Delta(\sigma_i),~\forall\, 1 \le i \le k \mbox{
    s.t. } |\sigma_i| \ge 2.
\end{align*}
Furthermore, since an error does not occur and the parties in $A$
attain omniscience, by Remark \ref{remark:error-event}
\begin{align*}
(R_j^{\mathtt{out}}: j \in A) \in \cR_{\mathtt{CO}}^\Delta(A).
\end{align*}
Thus, by Lemma \ref{l:completion_omn} and Lemma \ref{l:combine_omn},
\begin{align*}
(R_j^{\mathtt{out}} : j \in \sigma_{i_l}, 1 \le l \le c ) \in
  \cR_{\mathtt{CO}}^\Delta\left( \bigcup_{l=1}^c \sigma_{i_l} \right).
\end{align*}
Therefore, since no error has occurred, by Remark
\ref{remark:error-event} the parties in $\cup_{l=1}^c \sigma_{i_l}$
must attain omniscience. But by the definition of $\cO$ the set $A$ is
a maximal set attaining omniscience and $A \subseteq \cup_{l=1}^c
\sigma_{i_l}$.  Hence, $A$ must be $\cup_{l=1}^c \sigma_{i_l}$.

\paragraph*{Proof of (\ref{item:range-of-Rout})}
As a preparation of our proof, we first show that for $1 \le j,l \le
c$, the difference between $(R_{\sigma_{i_j}}^{\mathtt{out}} -
R_{\sigma_{i_j}}^*(A_\sigma))$ and $(R_{\sigma_{i_l}}^{\mathtt{out}} -
R_{\sigma_{i_l}}^*(A_\sigma))$ is bounded above by $(2 \alpha +
1)\Delta$. Indeed, for $1 \le j,l \le c$, the parties in
$\sigma_{i_j}$ and $\sigma_{i_l}$ are communicating when
$\mathrm{OMN}$ terminates. In fact, defining
\[
i_l \le s^\mathtt{in} := \max\{  {i :
  R_{\sigma_i}^{\mathtt{in}}} \ge 0\},
\]
the parties in $\sigma_{i_l}$ were communicating even when
$\mathrm{OMN}$ was initiated, and therefore,
\begin{align*}
R_{\sigma_1}^{\mathtt{out}} - R_{\sigma_{i_l}}^{\mathtt{out}} =
R_{\sigma_1}^{\mathtt{in}} - R_{\sigma_{i_l}}^{\mathtt{in}},
\end{align*} 
which by the assumption that $\mathbf{R}^{\mathtt{in}}$ is
$(\sigma,\mathbf{H},\alpha)$-valid yields
\begin{align} \label{eq:initially-communicating-parties}
H_{\sigma_1} - H_{\sigma_{i_l}} - \alpha \Delta \le
R_{\sigma_1}^{\mathtt{out}} - R_{\sigma_{i_l}}^{\mathtt{out}} \le
H_{\sigma_1} - H_{\sigma_{i_l}} + \alpha \Delta.
\end{align}
On the other hand, if $i_l > s^{\mathtt{in}}$, then the parties in
$\sigma_{i_l}$ start communicating when
\begin{align*}
R_{\sigma_1} = \lceil H_{\sigma_1} - H_{\sigma_{i_l}} + \alpha \Delta
\rceil_\Delta,
\end{align*}
where $\lceil a \rceil_\Delta := \min\{ i \Delta: i \in \mathbb{N},
i\Delta \ge a \}$.  Thereafter, the parties in $\sigma_1$ as well as
$\sigma_{i_l}$ communicate at sum-rate $\Delta$ per round. Thus, in
this case,
\begin{align} \label{eq:started-communicating-parties}
R_{\sigma_1}^{\mathtt{out}} - R_{\sigma_{i_l}}^{\mathtt{out}} = \lceil
H_{\sigma_1} - H_{\sigma_{i_l}} + \alpha \Delta \rceil_\Delta.
\end{align}
Upon combining \eqref{eq:initially-communicating-parties} and
\eqref{eq:started-communicating-parties}, we get that for every $1 \le
j,l \le c$,
\begin{align}
R_{\sigma_{i_j}}^{\mathtt{out}} - R_{\sigma_{i_l}}^{\mathtt{out}} &\le
H_{\sigma_{i_j}} - H_{\sigma_{i_l}} + (2\alpha +1) \Delta \nonumber
\\ &= R_{\sigma_{i_j}}^*(A_\sigma) - R_{\sigma_{i_l}}^*(A_\sigma) +
(2\alpha+1)\Delta,
 \label{eq:difference-within-2alpha+1}
\end{align} 
where the previous equation is by \eqref{e:rstar_diff_formula}.

Now, we prove the lower bound in (\ref{item:range-of-Rout}). Suppose
that there exists a $j \in \{1,\ldots,c\}$ such that
\begin{align*}
R_{\sigma_{i_j}}^{\mathtt{out}} < R_{\sigma_{i_j}}^*(A_\sigma) +
(|\sigma_{i_j}| - 2\alpha -1) \Delta.
\end{align*}
It follows from \eqref{eq:difference-within-2alpha+1} that
\begin{align}
\sum_{l=1}^c R_{\sigma_{i_l}}^{\mathtt{out}} &< \sum_{l=1}^c
R_{\sigma_{i_l}}^*(A_\sigma) + \sum_{l=1}^c |\sigma_{i_l}| \Delta
\nonumber \\ &= \sum_{l=1}^c R_{\sigma_{i_l}}^*(A_\sigma) + |A|\Delta
\nonumber \\ &= \mathbb{H}_{\sigma_f}(A_\sigma) + |A|
\Delta, \label{eq:upper-bound-for-contradiction}
\end{align}
where the previous equation is by \eqref{e:rstar_sum_formula2}. Also,
since no error occurs and parties in $A$ attain omniscience, by Remark
\ref{remark:error-event},
\begin{align*}
(R_j^{\mathtt{out}}: j \in A) \in \cR_{\mathtt{CO}}^\Delta(A),
\end{align*}
which in turn implies that
\begin{align*}
\sum_{l=1}^c R_{\sigma_{i_l}}^{\mathtt{out}} &= \frac{1}{c-1}
\sum_{l=1}^c \sum_{j=1 \atop j \neq l}^c
R_{\sigma_{i_j}}^{\mathtt{out}} \\ &\ge \frac{1}{c-1} \sum_{l=1}^c
\left[ H(X_A|X_{\sigma_{i_l}}) + (|A|- |\sigma_{i_l}|) \Delta \right]
\\ &= \mathbb{H}_{\sigma_f}(A_\sigma) + |A| \Delta,
\end{align*}
which contradicts \eqref{eq:upper-bound-for-contradiction}. Thus, for
every $1 \le l \le c$,
\begin{align*}
R_{\sigma_{i_l}}^{\mathtt{out}} &\ge R_{\sigma_{i_l}}^*(A_\sigma) +
(|\sigma_{i_l}| - 2 \alpha -1) \Delta \\ &\ge
R_{\sigma_{i_l}}^*(A_\sigma) - 2\alpha \Delta.
\end{align*}

Moving to the proof of the upper bound in (\ref{item:range-of-Rout}),
suppose that there exists an $l$ such that
\begin{align} \label{eq:assumption-contradiction}
R_{\sigma_{i_l}}^{\mathtt{out}} > R_{\sigma_{i_l}}^*(A_\sigma) + (m+
2\alpha +1)\Delta.
\end{align}
From Lemma \ref{lemma:induction}, there exists $B \subseteq
\{1,\ldots, c\}$ with $|B| \ge 2$ such that
\begin{align} \label{eq:exist-B}
(R_{\sigma_{i_l}}^*(A_\sigma) + |\sigma_{i_l}| \Delta : l \in B) \in
  \cR_{\mathtt{CO}}^\Delta(\{ \sigma_{i_l} : l \in B\}).
\end{align}
Then, \eqref{eq:difference-within-2alpha+1},
\eqref{eq:assumption-contradiction} and \eqref{eq:exist-B} imply (by
noting $|\sigma_{i_l}| < m$) that
\begin{align} \label{eq:Rout-included-in-RCO}
(R_{\sigma_{i_l}}^{\mathtt{out}} - \Delta : l \in B) \in
  \cR_{\mathtt{CO}}^\Delta(\{ \sigma_{i_l} : l \in B\}).
\end{align} 
Also, note that for every $j \in \sigma_{i_l}$, $1 \le l \le c$ with
$|\sigma_{i_l}| \ge 2$,
\begin{align} \label{eq:at-least-one-communication}
R_j^{\mathtt{out}} - R_j^{\mathtt{in}} \ge
\frac{\Delta}{|\sigma_{i_l}|},
\end{align}
since otherwise there is no communication in the execution of
$\mathrm{OMN}$, which in turn by Remark \ref{remark:error-event}
contradicts the assumption that $\mathbf{R}^{\mathtt{in}}$ is
$(\sigma, \mathbf{H},\alpha)$-valid. Upon combining
\eqref{eq:at-least-one-communication} with
\eqref{e:part_omn_condition}, we get
\begin{align*}
\left( R_j^{\mathtt{out}} - \frac{\Delta}{|\sigma_{i_l}|} : j \in
\sigma_{i_l} \right) \in \cR_{\mathtt{CO}}^\Delta(\sigma_{i_l})
\end{align*}
for every $1 \le l \le c$ with $|\sigma_{i_l}| \ge 2$, which together
with \eqref{eq:Rout-included-in-RCO} and Lemma \ref{l:combine_omn}
yields
\begin{align*}
\left( R_j^{\mathtt{out}} - \frac{\Delta}{|\sigma_{i_l}|} : j \in
\sigma_{i_l}, l \in B \right) \in \cR_{\mathtt{CO}}^\Delta\left(
\bigcup_{l \in B} \sigma_{i_l} \right).
\end{align*}
But then by Remark \ref{remark:error-event} the parties in $\cup_{l
  \in B} \sigma_{i_l}$ attain omniscience one round before
$\mathrm{OMN}$ terminates, which is a contradiction since no error has
occurred and $\mathrm{OMN}$ must terminate as soon as a subset in
$\cO$ is recognized.

\paragraph*{Proof of (\ref{item:validity-of-Rout})}
For each $\sigma_i^{\mathtt{out}} \in \sigma^{\mathtt{out}}$ either
$\sigma_i^{\mathtt{out}} \in \sigma$ or $\sigma_i^{\mathtt{out}} \in
\cO$; in the latter case, by (\ref{item:decomposition-of-A}),
$\sigma_i^{\mathtt{out}}$ must equal a union of parts of
$\sigma$. Note that, by the argument leading to
\eqref{eq:difference-within-2alpha+1}, for every $\sigma_i,\sigma_j
\in \sigma$ such that $R_{\sigma_i}^{\mathtt{out}} \ge 0$ and
$R_{\sigma_j}^{\mathtt{out}} \ge 0$
\begin{align} \label{eq:difference-sigma-i-sigma-j-out}
R_{\sigma_i}^{\mathtt{out}} - R_{\sigma_j}^{\mathtt{out}} \le
H_{\sigma_i} - H_{\sigma_j} + (2\alpha + 1)\Delta.
\end{align}
Also, for $\sigma_i^{\mathtt{out}} = \cup_{l=1}^c \sigma_{i_l} \in
\cO$, note that\footnote{We show the argument for $\sum_{l=2}^c
  R_{\sigma_{i_l}}^{\mathtt{out}}$; the same argument extends to
  $\sum_{l=1, l \neq i}^c R_{\sigma_{i_l}}^{\mathtt{out}}$ for every
  $i\in \{2, \ldots, c\}$.} by (\ref{item:range-of-Rout})
\begin{align*}
\sum_{l=2}^c R_{\sigma_{i_l}}^*(\{ \sigma_{i_1},\ldots,\sigma_{i_c}
\}) - 2 \alpha c \Delta &\le \sum_{l=2}^c
R_{\sigma_{i_l}}^{\mathtt{out}} \nonumber \\ &\le \sum_{l=2}^c
R_{\sigma_{i_l}}^*(\{ \sigma_{i_1},\ldots,\sigma_{i_c} \}) +
(mc+2\alpha c) \Delta,
\end{align*}
which by \eqref{e:rstar_sum_formula3} is the same as
\begin{align}
H(X_{\sigma_i^{\mathtt{out}}} | X_{\sigma_{i_1}}) - 2 \alpha c \Delta
&\le \sum_{l=2}^c R_{\sigma_{i_l}}^{\mathtt{out}} \nonumber \\ &\le
H(X_{\sigma_i^{\mathtt{out}}} | X_{\sigma_{i_1}}) + (mc + 2\alpha c)
\Delta.
  \label{eq:Rout-2-c-bound}
\end{align}

To prove condition (\ref{condition:constant-difference}) in the
definition of a valid rate vector ($cf.$~Definition \ref{d:valid}),
consider $\sigma_i^{\mathtt{out}}$ and $\sigma_j^{\mathtt{out}}$ such
that $R_{\sigma_i^{\mathtt{out}}}^{\mathtt{out}} \ge 0$ and
$R_{\sigma_j^{\mathtt{out}}}^{\mathtt{out}} \ge 0$. The following
three cases are possible:
\begin{itemize} 
\item Case $\sigma_i^{\mathtt{out}}, \sigma_j^{\mathtt{out}} \in
  \sigma$: In this case, the claim follows from
  \eqref{eq:difference-sigma-i-sigma-j-out}.

\item Case $\sigma_i^{\mathtt{out}} \in \cO$, $\sigma_j^{\mathtt{out}}
  \in \sigma$: Let $\sigma_i^{\mathtt{out}} = \cup_{l=1}^c
  \sigma_{i_l}$. Then, by \eqref{eq:difference-sigma-i-sigma-j-out}
  and \eqref{eq:Rout-2-c-bound}
\begin{align*}
R_{\sigma_i^{\mathtt{out}}}^{\mathtt{out}} -
R_{\sigma_j^{\mathtt{out}}}^{\mathtt{out}} &= \sum_{l=2}^c
R_{\sigma_{i_l}}^{\mathtt{out}} + R_{\sigma_{i_1}}^{\mathtt{out}} -
R_{\sigma_j^{\mathtt{out}}}^{\mathtt{out}} \\ &\le
H(X_{\sigma_i^{\mathtt{out}}} | X_{\sigma_{i_1}}) +
H(X_{\sigma_{i_1}}) - H(X_{\sigma_j^{\mathtt{out}}}) + (mc+2\alpha c +
2\alpha + 1) \Delta \\ &= H(X_{\sigma_i^{\mathtt{out}}}) -
H(X_{\sigma_j^{\mathtt{out}}}) + (mc+2\alpha c + 2\alpha + 1) \Delta,
\end{align*}
and similarly,
\begin{align*}
R_{\sigma_j^{\mathtt{out}}}^{\mathtt{out}} -
R_{\sigma_i^{\mathtt{out}}}^{\mathtt{out}} \le
H(X_{\sigma_j^{\mathtt{out}}}) - H(X_{\sigma_i^{\mathtt{out}}}) + (2
\alpha c + 2\alpha + 1)\Delta.
\end{align*} 

\item Case $\sigma_i^{\mathtt{out}}, \sigma_j^{\mathtt{out}} \in \cO$:
  Using argument similar to the previous case, we can show
\begin{align*}
R_{\sigma_i^{\mathtt{out}}}^{\mathtt{out}} -
R_{\sigma_j^{\mathtt{out}}}^{\mathtt{out}} \le
H(X_{\sigma_i^{\mathtt{out}}}) - H(X_{\sigma_j^{\mathtt{out}}}) + (mc
+ 4 \alpha c + 2\alpha + 1) \Delta.
\end{align*}
\end{itemize} 

Condition (\ref{condition:noncommunicating-parties}) can be proved
similarly by considering two cases: $\sigma_1^{\mathtt{out}} \in
\sigma$ and $\sigma_1^{\mathtt{out}} \in \cO$. Specifically, let
$s^\prime = \max\{ i : R_{\sigma_i}^{\mathtt{out}} \ge 0\}$.  If
$\sigma_1^{\mathtt{out}} \in \sigma$, then $\sigma_1^{\mathtt{out}} =
\sigma_1$ and condition (\ref{condition:noncommunicating-parties})
holds since the party\footnote{It can be seen that the parts of
  $\sigma$ which did not start communicating must be singleton.}
$\sigma_{s^\prime+1}^{\mathtt{out}}$ did not start communicating. On
the other hand, if $\sigma_1^{\mathtt{out}} = \cup_{l=1}^c
\sigma_{i_l} \in \cO$, then
\begin{align*}
R_{\sigma_1^{\mathtt{out}}}^{\mathtt{out}} &=
R_{\sigma_1}^{\mathtt{out}} + R_{\sigma_{i_1}}^{\mathtt{out}} -
R_{\sigma_1}^{\mathtt{out}} + \sum_{l=2}^c
R_{\sigma_{i_l}}^{\mathtt{out}} \\ &< H_{\sigma_1} -
H_{\sigma_{s^\prime +1}^{\mathtt{out}}} + H_{\sigma_{i_1}} -
H_{\sigma_1} + H(X_{\sigma_1^{\mathtt{out}}} | X_{\sigma_{i_1}}) + (mc
+ 2 \alpha c + 3\alpha +1) \Delta \\ &= H_{\sigma_1^{\mathtt{out}}} -
H_{\sigma_{s^\prime +1}^{\mathtt{out}}} + (mc + 2 \alpha c + 3\alpha
+1) \Delta,
\end{align*}
where the strict inequality is by
\eqref{eq:difference-sigma-i-sigma-j-out} and
\eqref{eq:Rout-2-c-bound}, since the party
$\sigma_{s^\prime+1}^{\mathtt{out}}$ did not start communicating.

For condition (\ref{condition:combined-parties}), if
$\sigma_i^{\mathtt{out}}$ equals to a part of $\sigma$, then
$(R_j^{\mathtt{out}} : j \in \sigma_i^{\mathtt{out}}) \in
\cR_{\mathtt{CO}}^\Delta(\sigma_i^{\mathtt{out}})$ since
$(R_j^{\mathtt{in}} : j \in \sigma_i^{\mathtt{out}}) \in
\cR_{\mathtt{CO}}^\Delta(\sigma_i^{\mathtt{out}})$ and
$R_j^{\mathtt{out}} \ge R_j^{\mathtt{in}}$ for every $1 \le j \le
m$. On the other hand, if $\sigma_i^{\mathtt{out}}$ belongs to $\cO$,
then $(R_j^{\mathtt{out}} : j \in \sigma_i^{\mathtt{out}}) \in
\cR_{\mathtt{CO}}^\Delta(\sigma_i^{\mathtt{out}})$ by Remark
\ref{remark:error-event}.

Finally, for condition (\ref{condition:separate-parts}), if there
exists $A \subseteq \{1,\ldots, |\sigma^{\mathtt{out}}| \}$, $|A| \ge
2$, such that
\begin{align*}
(R_j^{\mathtt{out}} : j \in \sigma_i^{\mathtt{out}}, i \in A) \in
  \cR_{\mathtt{CO}}^\Delta\left( \bigcup_{i\in A}
  \sigma_i^{\mathtt{out}} \right),
\end{align*} 
then by Remark \ref{remark:error-event} $\cup_{i \in A}
\sigma_i^{\mathtt{out}} \in \cO$, which further implies that $\cup_{i
  \in A} \sigma_i^{\mathtt{out}} \in \cO$ is a part of
$\sigma^{\mathtt{out}}$, a contradiction. Thus, condition
(\ref{condition:separate-parts}) must hold for
$\mathbf{R}^{\mathtt{out}}$. \qed

%%%%%%%%%%%%%

\subsection{Proofs of Theorem~\ref{theorem:for-individual-sequence}  and Corollary \ref{col:main}}

For Theorem~\ref{theorem:for-individual-sequence}, by Lemma
\ref{lemma:decoding-error} the probability of the error event $\cE =
\cE(\mathbf{x})$ is bounded above by $C_1 L p(n) 2^{-n \Delta}$ for
some constant $C_1$, where $L$ is the maximum number of rounds and is
bounded above by $\frac{\log |\cX_{\cM}|}{\Delta}+m$ using
\eqref{e:L_bound}. Under the assumption that the error event $\cE$ did
not occur, at the end of the $j$th call to $\omn$ with input partition
$\sigma$, Theorem \ref{t:recursion} guarantees that the total number
of bits sent by each subset $A \in \cO$ is bounded above
by\footnote{The $\log n$ term corresponds to the bits communicated to
  share types of the locally recovered observations. Additional $C_3
  L$ bits are added to account for the overhead arising from
  rounding-off the required number of bits to an integer and ACK/NACK
  bits for each round.}
\begin{align} \label{eq:total-number-each-call-OMN}
n \mH_{\sigma_f(A_\sigma)}(A_\sigma |\bPP{\bx_A}) + n c (m+2 \alpha_j)
\Delta + C_3 L + C_4 \log n,
\end{align}
for some constants $C_3,C_4 > 0$, where $\alpha_j$ is recursively
defined by setting $\alpha_1 =1$ and $\alpha_{j+1} = c_m^\prime
\alpha_j$ with $c_m^\prime$ given in
Theorem~\ref{t:recursion}-(\ref{item:validity-of-Rout}). Since the
size of partition $\sigma$ strictly decreases in each execution of
$\omn$, the number of calls to $\omn$ is at most $m$ and $\alpha_j$
remains bounded above by a constant that depends only on $m$.
Theorem~\ref{theorem:for-individual-sequence} follows upon using
\eqref{eq:total-number-each-call-OMN} for $A = \cM$, and noting that
\begin{align*}
\mH_{\sigma_f(\cM_\sigma)}(\cM_\sigma |\bPP{\bx_\cM}) &=
\mH_{\sigma}(\cM | \bPP{\bx_\cM}) \\ &\le
R_{\mathtt{CO}}(\cM|\bPP{\mathbf{x}_\cM}),
\end{align*}
where the inequality is by \eqref{e:rco_formula}.

Corollary~\ref{col:main} is obtained as a consequence of
Theorem~\ref{theorem:for-individual-sequence} as follows.  First, note
that under the error event $\cE$, which occur with probability less
than $C_1 p(n) L 2^{-n \Delta}$, the number of communicated bits is
bounded above by $C_5 n$ for some constant $C_5>0$.  Next, by the
Taylor approximation of the entropy function around $\bPP{X_\cM}$, for
$\bQQ{X_\cM}$ satisfying $\| \bPP{X_\cM} - \bQQ{X_\cM} \| \le \delta$
and $\mathtt{supp}(\bQQ{X_\cM}) \subset \mathtt{supp}(\bPP{X_\cM})$,
we have
\begin{align} \label{eq:taylor}
\big| R_{\mathtt{CO}}(\cM| \bPP{X_\cM}) - R_{\mathtt{CO}}(\cM|
\bQQ{X_\cM}) \big| \le C_6 \delta
\end{align}
for a sufficiently small $\delta$, where $C_6 >0$ is a constant that
depends\footnote{The dependence of $C_6$ on $\bPP{X_\cM}$ can be
  omitted by replacing $\delta$ with $\delta\log
  \frac{|X_\cM|}{\delta}$.} on $\bPP{X_\cM}$.  Denoting
\begin{align*}
\cB_\delta(\bPP{X_\cM}) := \{ \bQQ{X_\cM} : \| \bPP{X_\cM} -
\bQQ{X_\cM} \| \le \delta,\, \mathtt{supp}(\bQQ{X_\cM}) \subset
\mathtt{supp}(\bPP{X_\cM}) \},
\end{align*}
Theorem~\ref{theorem:for-individual-sequence} implies that, when $\cE$
does not occur, the number of bits communicated is no more than
\begin{align*}
&nR_{\mathtt{CO}}(\cM|\bPP{X_\cM}) + C_6 n \delta + \Pr\big(
  \mathtt{type}(X_\cM^n) \notin \cB_\delta(\bPP{X_\cM}) \big) n \log
  |\cX_\cM| + C_2 n\Delta + C_3 L + C_4 \log n \\ &\leq n
  R_{\mathtt{CO}}(\cM|\bPP{X_\cM}) + C_6 n \delta + 2|\cX_\cM|
  \exp(-2n\delta^2) n \log |\cX_\cM| + C_2 n \Delta + C_3 L + C_4 \log
  n,
\end{align*}
where the inequality uses the Hoeffding bound
\begin{align*}
\Pr\big( \mathtt{type}(X_\cM^n) \notin \cB_\delta(\bPP{X_\cM}) \big)
\le 2 |\cX_\cM| \exp(-2n\delta^2).
\end{align*}
The claimed upper bound for the expected number of bits communicated
follows by combining the bounds under $\cE$ and $\cE^c$ and setting
$\delta = \sqrt{\frac{ \log n}{n}}$, $\Delta =
\frac{1}{\sqrt{n}}$. \qed

%%%%%%%%%%%%%%%%%%%%%
\subsection{Proof of Theorem~\ref{t:performance_SK_protocol}}

We first recall the leftover hash lemma ($cf.$~\cite{Ren05}); a proof
of the version stated below is given in, for instance, \cite[Appendix
  B]{HayTW14}.

\begin{lemma}[Leftover Hash] \label{lemma:leftover-hash}
Consider random variables $X$ and $V$ taking values in finite sets
$\cX$ and $\cV$, respectively. Let $S$ be a random seed such that
$f_S$ is uniformly distributed over a $2$-universal hash family. Then,
for $K = f_S(X)$, we have
\begin{align*}
\| \bPP{K V S} - \bPP{\mathtt{unif}} \times \bPP{V} \times \bPP{S}
\|_1 \le \frac{1}{2} \sqrt{ |\cV| 2^{- H_{\min}(\bPP{X})}},
\end{align*}
where $\bPP{\mathtt{unif}}$ is the uniform distribution on $\cK$ and
\begin{align*}
H_{\min}(\bPP{X}) = - \log \max_x \bP{X}{x}.
\end{align*}
\end{lemma}

We assume that the public randomness $U$ used in Protocol~\ref{p:main}
is available to the eavesdropper.
%% Let $U$ be the randomness used in
%% Protocol~\ref{p:main}.\footnote{Protocol~\ref{p:main} only uses the 
%%   public randomness; we assume that the public randomness is also 
%%   known to the eavesdropper.}  
Denote by $\cE$ the error event of Protocol~\ref{p:main}, which is
determined by $(X_\cM^n,U)$, and by $\Pi^\prime$ an expurgated
transcript defined as
\begin{align*}
\Pi^\prime = \left\{
\begin{array}{ll}
\Pi, & \mbox{if } (X_\cM^n,U) \notin \cE, \\ \mathtt{constant}, &
\mbox{otherwise.}
\end{array}
\right.
\end{align*}
Our security analysis will show that $\Pi^\prime$ reveals negligible
information about the SK and then use the large probability of
agreement between $\Pi$ and $\Pi^\prime$ to claim the security of the
SK. Note that when the joint type of $X_\cM^n$ is
$\bPP{\overline{X}_\cM}$ and an error did not occur in
Protocol~\ref{p:main}, the length of the transcript $\Pi$ is bounded
by $l(\bPP{\overline{X}_\cM})$; thereby the length of $\Pi^\prime$ is
bounded by $l(\bPP{\overline{X}_\cM})$ as well.

For each realization $X_\cM^n = \mathbf{x}$, we generate a SK of
length $k(\bPP{\mathbf{x}})$ by randomly hashing $\mathbf{x}$ to
$k(\bPP{\mathbf{x}})$ bits\footnote{Specifically, we use a seeded
  extractor for each fixed joint type $\bPP{\mathbf{x}}$. For ease of
  presentation, we omit the dependence on seed from our
  notation.}. Clearly, the recoverability condition is satisfied with
$1-\ep_n$, where $\ep_n$ is the error probability of
Protocol~\ref{p:main}. Without loss of generality, we assume that the
eavesdropper has access to the joint type $\bPP{\mathbf{x}}$ of
$\mathbf{x}$. Note that such an eavesdropper has potentially more
information than that available to the actual eavesdropper in our
protocol. Thus, security against this stronger eavesdropper implies
security against the actual eavesdropper. Denoting by $T=t$ a fixed
realization of the random type, triangular inequality yields
\begin{align*}
& \sum_{t \in \cP_n(\cX_\cM)} P_T(t) \| \bPP{K \Pi U |T=t} -
  \bPP{\mathtt{unif},t} \times \bPP{\Pi U|T=t} \| \\ &~~~\le \sum_{t
    \in \cP_n(\cX_\cM)} P_T(t) \bigg[ \| \bPP{K\Pi U |T=t} - \bPP{K
      \Pi^\prime U |T=t} \| + \| \bPP{\Pi^\prime U|T=t} - \bPP{\Pi
      U|T=t} \| \\
%\hspace{5cm}
&~~~~~~+ \| \bPP{K\Pi^\prime U|T=t} - \bPP{\mathtt{unif}} \times
    \bPP{\Pi^\prime U|T=t} \| \bigg],
\end{align*}
where $\bPP{\mathtt{unif},t}$ is the uniform distribution on
$\{0,1\}^{k(t)}$.  The first two terms on the right-side above are
each bounded above by $\Pr\big( \Pi \neq \Pi^\prime \big)$. Also, by
Lemma \ref{lemma:leftover-hash} applied for each fixed $(t,u)$, the
third term is bounded above by
\[
 \sum_{t \in \cP_n(\cX_\cM)} P_T(t) \frac{1}{2} \sqrt{ 2^{l(t) + k(t)}
   (n+1)^{|\cX_\cM|} 2^{- H(t)}},
\]
where we have used the independence of $U$ and $X_\cM^n$ and the
observation that
\begin{align*}
H_{\min}(\bPP{X_\cM^n|T=t,U=u}) &= H_{\min}(\bPP{X_\cM^n|T=t}) \\ &\ge
n H(t) - |\cX_\cM| \log (n+1).
\end{align*}
Thus, by combining the bounds above, we get
\begin{align*}
\lefteqn{ \sum_{t \in \cP_n(\cX_\cM)} P_T(t) \| \bPP{K \Pi U |T=t} -
  \bPP{\mathtt{unif},t} \times \bPP{\Pi U|T=t} \| } \\ &\le 2 \Pr\big(
\Pi \neq \Pi^\prime \big) + \sum_{t \in \cP_n(\cX_\cM)} P_T(t)
\frac{1}{2} \sqrt{ 2^{l(t) + k(t)} (n+1)^{|\cX_\cM|} 2^{- H(t)}}
\\ &\le \delta,
\end{align*}
where the previous inequality uses $\bPr{\Pi \neq \Pi^\prime} \leq
\ep_n$ and the definitions of $l(t)$, $k(t)$, and $\delta$.  The
average length $ \sum_{t \in \cP_n(\cX_\cM)} P_T(t) k(t)$ is lower
bounded by \eqref{eq:lower-bound-average-key-length} using
Theorem~\ref{theorem:key-capacity}, in a similar manner as the proof
of Corollary \ref{col:main}. \qed

%%%%%%%%%%%%%%%%%%%%%%%%%%%%%%%%%%%%%%%%%%%%%%%
%% \section*{Appendix A: Converse for general protocols}

%% %%%%%%%%%%%%%%%%%%%%%%%%%%%%%%%%%%%%%%%%%%%%%%%
%% \section*{Appendix B: Analysis for the ideal protocol}

%%%%%
\section*{Acknowledgment}

SW is supported in part by the JSPS KAKENHI under grant 16H06091. HT
is supported in part by the Defense Research and Development
Organisation (DRDO), India under grant DRDO0649. Authors thank Navin
Kashyap for pointing to \cite[Theorem 5.2]{ChanBEKL15} and
\cite{ChanB15}.

%%%%%%%%%%%%%%%%%%%%%%%%%%%%%%%%%%%%%%%%%%%%%%%
\bibliography{IEEEabrv,references} \bibliographystyle{IEEEtranS}

 %%%%%%%%%%%%%%%%%%%%%%%%%%%%%%%%%%%%%%%%%
\end{document}